\documentclass[11pt,english,plain,noend]{article}
\usepackage[T1]{fontenc}
\usepackage[latin9]{inputenc}
\usepackage{geometry}
\usepackage{babel}
\usepackage{booktabs}
\usepackage{multirow}
\usepackage{algorithm2e}
\usepackage{amsmath}
\usepackage{amsthm}
\usepackage{amssymb}
\usepackage{graphicx}
\usepackage{setspace}
\usepackage[authoryear]{natbib}
\usepackage[unicode=true,
 bookmarks=false,
 breaklinks=false,pdfborder={0 0 1},backref=section,colorlinks=false]
 {hyperref}

\makeatletter

\providecommand{\tabularnewline}{\\}

\theoremstyle{plain}
\newtheorem{thm}{\protect\theoremname}
\theoremstyle{definition}
\newtheorem{example}[thm]{\protect\examplename}

\@ifundefined{date}{}{\date{}}
\usepackage{babel}

\allowdisplaybreaks

\newtheorem{assumption}{Assumption}\newtheorem{theorem}{Theorem}\newtheorem{corollary}{Corollary}\newtheorem{lemma}{Lemma}

\newtheorem{step}{Step}\newtheorem{remark}{Remark}

\newcommand{\bbR}{\mathbb{R}}

\newcommand{\E}{\mathrm{E}}
\newcommand{\pr}{\mathrm{pr}}
\newcommand{\var}{\mathrm{var}}

\newcommand{\bfX}{{X}}
\newcommand{\bfx}{{x}}
\newcommand{\bfu}{{u}}

\newcommand{\vctl}{\mathrm{vecl}}
\newcommand{\csp}{\mathrm{span}}
\newcommand{\cond}{\mid}
\newcommand{\T}{ \mathrm{\scriptscriptstyle T}}
\newcommand{\indep}{\;\, \rule[0em]{.03em}{.67em} \hspace{-.25em}
	\rule[0em]{.65em}{.03em} \hspace{-.25em}
	\rule[0em]{.03em}{.67em}\;\,}

\newcommand{\bbP}{\mathbb{P}}

\providecommand{\examplename}{Example}
\providecommand{\theoremname}{Theorem}

\makeatother

\providecommand{\examplename}{Example}
\providecommand{\theoremname}{Theorem}

\begin{document}
\title{\textbf{\huge{}{}{}Robust inference of conditional average treatment
effects using dimension reduction }}
\author{Ming-Yueh Huang\thanks{Institute of Statistical Science, Academia Sinica, Taipei 11529, Taiwan
China. Email: myh0728@stat.sinica.edu.tw}$\ $ and Shu Yang\thanks{Department of Statistics, North Carolina State University, North Carolina
27695, U.S.A. Email: syang24@ncsu.edu.}}

\maketitle
\bigskip{}

\noindent \begin{flushleft}
\textbf{Abstract}: It is important to make robust inference of the
conditional average treatment effect from observational data, but
this becomes challenging when the confounder is multivariate or high-dimensional.
In this article, we propose a double dimension reduction method, which
reduces the curse of dimensionality as much as possible while keeping
the nonparametric merit. We identify the central mean subspace of
the conditional average treatment effect using dimension reduction.
A nonparametric regression with prior dimension reduction is also
used to impute counterfactual outcomes. This step helps improve the
stability of the imputation and leads to a better estimator than existing
methods. We then propose an effective bootstrapping procedure without
bootstrapping the estimated central mean subspace to make valid inference. 
\par\end{flushleft}

\textit{Key words}: augmented inverse probability weighting; matching;
kernel smoothing; U-statistic; weighted bootstrap.

\pagebreak{}

\section{Introduction\label{sec:intro}}

In recent biomedical and public health research, there has been growing
interest in developing valid and robust inference methods for the
conditional average treatment effect, which is also known as the treatment
contrast or heterogeneity of treatment effect. In particular, the
sign of conditional average treatment effects can be used to determine
the optimal individualized treatment regime \citep{QianM2011Performance}.
In the era of big data, large observational data with multivariate
or high-dimensional confounders are becoming increasingly available
for research purposes, such as electronic health records, claims,
and disease registries. It becomes challenging to make robust inference
of the conditional average treatment effect due to the curse of dimensionality,
calling for new techniques such as dimension reduction.

A large body of literature focuses on modeling the prognostic scores,
defined as the outcome mean functions under the treated and control
conditions. Parametric approaches include likelihood-based methods
\citep{ThallMS2000Evaluating,ThallSE2002Selecting,ThallWLMT2007Bayesian}
and parametric Q-learning \citep{ChakrabortyMS2010Inference}. Machine
learning-based methods include Bayesian additive regression trees
\citep{ChipmanGM2010Bayesian}, causal multivariate adaptive regression
splines \citep{PowersQJSSHT2018Some}, random forest, causal boosting
trees, and reinforcement learning \citep{ZhaoZSK2011Reinforcement}.
Although the conditional average treatment effect is simply the difference
between the treated and the control prognostic scores, they may have
different model features. Thus, in some applications, direct modeling
on the conditional average treatment effect may provide a more accurate
characterization of treatment effects, avoiding redundancy of non-useful
features. Another body of literature focuses on modeling and approximating
the conditional average treatment effect both parametrically \citep{Murphy2003Optimal,Robins2004Optimal}
and using machine learning methods \citep{ZhaoZRK2012Estimating,ZhangTDZL2012Estimating,RzepakowskiJ2012Decision,AtheyI2016Recursive,AtheyTW2019Generalized}.
However, parametric methods are susceptible to model misspecification.
Although machine learning is flexible, it often produces results that
are too complicated to be interpretable. Most importantly, it is a
daunting task to draw valid inference based on machine learning methods.
Semiparametric methods offer compromises between fully parametric
and machine learning approaches. \citet{SongLZZLL2017Semiparametric}
and \citet{LiangY2020Semiparametric} considered single index and
multiple index models for the treatment contrast function, respectively.
See also \citet{LuoWZ2019Learning}.

In this article, we propose a nonparametric framework to make robust
inference of the conditional average treatment effect. To mitigate
the possible curse of dimensionality, we consider the central mean
subspace of the conditional average treatment effect, which is the
smallest linear subspace spanned by a set of linear index that can
sufficiently characterize the estimand of interest \citep{CookL2002Dimension}.
Precisely speaking, a series of nested multiple index models are considered
and the most saturated model is the fully nonparametric regression.
Any conditional average treatment effect must belong to this series
with a particular number of linear indices, which is called the structural
dimension. The primary goal is to estimate this structural dimension
and the corresponding index coefficients. Under this framework, we
specify the conditional average treatment effect nonparametrically
and use a model selection procedure to determine a sufficient structural
dimension.

To estimate the central mean subspace, we propose imputing counterfactual
outcomes by kernel regression with a prior dimension reduction. The
prior dimension reduction helps to improve the stability of the imputation
and the subsequent estimation of the conditional average treatment
effect. It is worth comparing our imputation approach to alternative
methods. For example, nearest neighbor imputation can be used to impute
the missing potential outcome by pairing each unit with the nearest
neighbor in the opposite treatment group. However, matching is generally
not effective in the presence of many covariates. \citet{AbrevayaHL2015Estimating}
considered an inverse probability weighted adjusted outcome for the
conditional average treatment effect. It is well known that weighted
by the inverse probabilities is highly unstable \citep{KangS2007Demystifying}.
To overcome this instability, \citet{ZhaoZRK2012Estimating} further
considered combining weighting and outcome regression as an augmented
inverse probability adjusted outcome. In our simulation study, regression
imputation and augmented inverse probability weighting have comparable
performances. Because the inverse weighting methods require additional
propensity score estimation, we suggest regression imputation to save
more computational time in practice.

Theoretically, we derive the consistency and asymptotic normality
of the proposed estimator of the conditional average treatment effect.
The main challenge is that the imputed counterfactual outcomes are
not independent. To overcome this challenge, we calculate the difference
between\textbf{ }imputed and conditional counterfactual outcomes,
which can be expressed as a weighted empirical average of the influence
functions of the kernel regression estimator. Thus, we can show that
the influence function of the proposed estimator can be approximated
by a U-statistic. Invoking the properties of degenerate U-processes
discussed in \citet{NolanP1987U}, we can derive the asymptotic distribution
of the estimated conditional average treatment effect and show that
the imputation step plays a non-negligible role. To make valid inference,
we propose a under-smooth strategy such that the asymptotic bias is
dominated by the asymptotic variance. We can estimate the asymptotic
variances by applying weighted bootstrap techniques and construct
Wald confidence intervals. Interestingly, the fact that the central
mean subspace is estimated does not affect the asymptotic distribution
of the proposed estimator of the conditional average treatment effect.
Thus, in our bootstrap procedure, we can safely skip the step of bootstrapping
the estimated central mean subspace, which saves a lot of computation
time in practice.


\section{Methodology\label{sec:metho}}

\subsection{Preliminaries\label{subsec:pre}}

We use potential outcomes to define causal effects. Suppose that the
binary treatment is $A\in\{0,1\}$, with $0$ and $1$ being the labels
for the control and active treatment, respectively. Each level of
treatment $a$ corresponds to a potential outcome $Y(a)$, representing
the outcome had the subject, possibly contrary to the fact, been given
treatment $a$. The individual causal effect is $D=Y(1)-Y(0)$. Let
$\bfX\in\mathbb{R}^{p}$ be a $p$-vector of pre-treatment covariates.
The covariate-specific average treatment effect is $\tau(\bfx)=\E\{Y(1)-Y(0)\cond\bfX=\bfx\}=\E(D\cond\bfX=\bfx)$.
The observed outcome is $Y=Y(A)=AY(1)+(1-A)Y(0)$. The main goal of
this article is to estimate $\tau(\bfx)$ based on observational data
$\{(A_{i},Y_{i},\bfX_{i}):i=1,\dots,n\}$, which independently and
identically follows $f(A,Y,\bfX)$.

To identify the treatment effects, we assume the following assumptions,
which are standard in causal inference with observational studies
\citep{RosenbaumR1983Central}:

\begin{assumption}\label{assump:ignorable} $\{Y(0),Y(1)\}\indep A\cond\bfX$.\end{assumption}

\begin{assumption}\label{asump:overlap} There exist constants $c_{1}$
and $c_{2}$ such that $0<c_{1}\leq\pi(\bfX)\leq c_{2}<1$ almost
surely, where $\pi(\bfx)=\pr(A=1\cond\bfX=\bfx)$ is the propensity
score. \end{assumption}

Let $\mu_{a}(\bfx)=\E\{Y(a)\cond\bfX=\bfx\}$ ($a=0,1$). Under Assumptions
\ref{assump:ignorable}--\ref{asump:overlap}, $\mu_{a}(\bfx)=\E(Y\cond A=a,\bfX=\bfx)$
and $\tau(\bfx)=\mu_{1}(\bfx)-\mu_{0}(\bfx)$ are identifiable from
$f(A,Y,\bfX)$. This identification formula motivates a common strategy
of estimating $\tau(\bfx)$ by approximating $\mu_{a}(X)$ separately
for $a=0,1$. Alternatively, we propose robust inference of $\tau(\bfx)$
directly using dimension reduction, which requires no parametric model
assumptions and can detect accurate and parsimonious structures of
$\tau(\bfx)$.

\subsection{Dimension reduction on conditional average treatment effect}

The main idea is to search for the fewest linear indices $B_{\tau}^{\T}\bfx$
such that 
\begin{align}
\tau(\bfx)=g(B_{\tau}^{\T}\bfx),\label{eq:SDM}
\end{align}
where $B_{\tau}$ is a $p\times d_{\tau}$ matrix consisting of index
coefficients, and $g$ is an unknown $d_{\tau}$-variate function.
Since $\tau(\bfx)=\E(D\cond\bfX=\bfx)$, the column space of $B_{\tau}$
is called the central mean subspace of $D$ given $\bfX$, denoted
by $\mathcal{S}_{\E(D|\bfX)}$ \citep{CookL2002Dimension}.

The central mean subspace $\mathcal{S}_{\E(D|\bfX)}$ is nonparametric.
In other words, for any multivariate function $\tau(\bfx)$, without
particular parametric or semiparametric modeling, there always exists
a unique central mean subspace. To illustrate, consider the single-index
model $g(\bfx^{\T}\beta)$ which leads to a one-dimensional central
mean subspace spanned by $\beta$. Unlike the single-index model that
prefixes the dimension of the central mean subspace, we leave both
$d_{\tau}$ and $B_{\tau}$ unspecified, and the primary goal of dimension
reduction is to estimate $d_{\tau}$ and $B_{\tau}$. In addition,
the curse of dimensionality can be avoided if $d_{\tau}$ is much
smaller than $p$.

\begin{remark} Recall that $\tau(\bfx)=\mu_{1}(\bfx)-\mu_{0}(\bfx)$.
An alternative way to employ dimension reduction is to search for
two sets of linear indices $B_{0}^{\T}\bfx$ and $B_{1}^{\T}\bfx$
such that 
\begin{align}
\mu_{0}(\bfx)=g_{0}(B_{0}^{\T}\bfx),\qquad\mu_{1}(\bfx)=g_{1}(B_{1}^{\T}\bfx),\label{eq:outcomeMs}
\end{align}
where $g_{0}$ and $g_{1}$ are unknown functions. That is, we can
also estimate $\mathcal{S}_{\E\{Y(0)|\bfX\}}=\csp(B_{0})$ and $\mathcal{S}_{\E\{Y(1)|\bfX\}}=\csp(B_{1})$,
and then recover $\tau(\bfx)$ by $g_{1}(B_{1}^{\T}\bfx)-g_{0}(B_{0}^{\T}\bfx)$.
In fact, we can show that $\mathcal{S}_{\E(D|\bfX)}\subseteq\mathcal{S}_{\E\{Y(0)|\bfX\}}+\mathcal{S}_{\E\{Y(1)|\bfX\}}$,
where the sum of two linear subspaces is $U+V=\{u+v:u\in U,v\in V\}$.
In some cases $\mathcal{S}_{\E(D|X)}$ may have a strictly smaller
dimension than $\mathcal{S}_{\E\{Y(0)|\bfX\}}$ and $\mathcal{S}_{\E\{Y(1)|\bfX\}}$
as demonstrated by the following example. Thus, using model (\ref{eq:SDM})
may detect more parsimonious structures of $\tau(\bfx)$ than using
model (\ref{eq:outcomeMs}).

\end{remark} 
\begin{example}
Let $Y(0)=\alpha^{\T}\bfX+(\beta^{\T}\bfX)^{2}$ and $Y(1)=\alpha^{\T}\bfX+(\beta^{\T}\bfX)^{3}$,
where $\alpha,\beta\in\mathbb{R}^{p}$, and $\alpha$ and $\beta$
are not linearly dependent. Then, $\dim(\mathcal{S}_{\E\{Y(0)|\bfX\}})=\dim(\mathcal{S}_{\E\{Y(1)|\bfX\}})=\dim\{\csp(\alpha,\beta)\}=2$,
while $\dim(\mathcal{S}_{\E(D|\bfX)})=\dim\{\csp(\beta)\}=1$. 
\end{example}

\begin{remark} As discussed in \citet{MaZ2013Efficient}, the parameter
$B$ is not identifiable without further restrictions. To see this,
suppose that $Q$ is an invertible $d\times d$ matrix and consider
$g^{\ast}(\bfu)=g\{(Q^{\T})^{-1}\bfu\}$. Then we can derive another
equivalent representation of $\tau(x)$ as 
\begin{align*}
\tau(\bfx)=g(B^{\T}\bfx)=g\{(Q^{\T})^{-1}Q^{\T}B^{\T}\bfx\}=g^{\ast}\{(BQ)^{\T}\bfx\}.
\end{align*}
Thus, the two sets of parameters $(B,g)$ and $(BQ,g^{\ast})$ correspond
to the same conditional average treatment effect. As a result, the
central subspace was introduced to make the column space invariant
to these invertible linear transformations. We use a particular parametrization
of the central mean subspace as used in \citet{MaZ2013Efficient}.
Without loss of generality, we set the upper $d\times d$ block of
$B$ to be the identity matrix $I_{d\times d}$ and write $\bfX=(\bfX_{u}^{\T},\bfX_{l}^{\T})^{\T}$,
where $\bfX_{u}\in\bbR^{d}$ and $\bfX_{l}\in\bbR^{p-d}$. Hence,
the free parameters are the lower $(p-d)\times d$ entries of $B$,
corresponding to the coefficients of $\bfX_{l}$. For generic matrix
$B$, we now denote $\vctl(B)$ as the vector formed by the lower
$(p-d)\times d$ entries of $B$. \end{remark}

\subsection{Imputation and Estimation\label{subsec:Imputation-and-Estimation}}

If $D$ were known, existing methods can be directly applied to estimate
$\mathcal{S}_{\E(D|X)}$. However, the fundamental problem in causal
inference is that the two potential outcomes can never be jointly
observed for each unit, one is factual $Y(A)$ and the other one is
counterfactual $Y(1-A)$. To overcome this challenge, we propose an
imputation step to impute the counterfactual outcomes. A natural choice
to impute $Y(1-A)$ is using its conditional mean given $X$, $\mu_{1-A}(\bfX)$.
As mentioned in $\mathsection$ \ref{subsec:pre}, $\mu_{a}(\bfx)$
can be estimated by matching or other nonparametric smoothing techniques.
To further reduce the possible curse of dimensionality, we propose
a prior dimension reduction procedure to estimate $\mu_{a}(\bfx)$.

The proposed imputation and estimation procedure proceeds as follows.

\begin{step} Estimate the central mean subspace $\mathcal{S}_{\E\{Y(a)|\bfX\}}$
($a=0,1$). Let $\mu_{a}(\bfu;B)=\E(Y\cond A=a,B^{\T}\bfX=\bfu)$,
where $B$ is a $p\times d$ parameter matrix. Given $B$, the kernel
smoothing estimator of $\mu_{a}(\bfu;B)$ is 
\begin{align}
\widehat{\mu}_{a}(\bfu;B)=\frac{\sum_{j=1}^{n}Y_{j}1(A_{j}=a)\mathcal{K}_{q,h}(B^{\T}\bfX_{j}-\bfu)}{\sum_{j=1}^{n}1(A_{j}=a)\mathcal{K}_{q,h}(B^{\T}\bfX_{j}-\bfu)},\label{eq:mu-hat-B}
\end{align}
where $1(\cdot)$ is the indicator function, $\mathcal{K}_{q,h}(\bfu)=\prod_{k=1}^{d}K_{q}(u_{k}/h)/h$
with $\bfu=(u_{1},\dots,u_{d})$, $K_{q}$ is a $q$th ordered and
twice continuously differentiable kernel function with bounded support,
and $h$ is a positive bandwidth. The basis matrix of $\mathcal{S}_{\E\{Y(a)|\bfX\}}$
can be estimated by $\widehat{B}_{a}$, where $(\widehat{d}_{a},\widehat{B}_{a},\widehat{h}_{a})$
is the minimizer of the cross-validation criterion 
\begin{align}
\textsc{cv}_{a}(d,B,h)=\sum_{i=1}^{n}\{Y_{i}-\widehat{\mu}_{a}^{-i}(B^{\T}\bfX_{i};B)\}^{2}1(A_{i}=a),\label{eq:cv}
\end{align}
where the superscript $-i$ indicates the estimator (\ref{eq:mu-hat-B})
based on data without the $i$th subject. This criterion (\ref{eq:cv})
is a mean regression version of \citet{HuangC2017Effective}. In the
optimization, the order of the kernel function $q>\max(d/2+1,2)$
is specified for each working dimension $d$. \end{step}

\begin{step} Impute the individual treatment effect by 
\begin{align*}
\widehat{D}_{i}=A_{i}\{Y_{i}-\widehat{\mu}_{0}(\widehat{B}_{0}^{\T}\bfX_{i};\widehat{B}_{0})\}+(1-A_{i})\{\widehat{\mu}_{1}(\widehat{B}_{1}^{\T}\bfX_{i};\widehat{B}_{1})-Y_{i}\}\enspace(i=1,\dots,n)
\end{align*}
with specified orders $(q_{0},q_{1})$ of kernel functions and bandwidths
$(h_{0},h_{1})$ in $\widehat{\mu}_{0}(\widehat{B}_{0}^{\T}\bfX_{i};\widehat{B}_{0})\}$
and $\widehat{\mu}_{1}(\widehat{B}_{1}^{\T}\bfX_{i};\widehat{B}_{1})$.
The choices of $q_{0}$ and $q_{1}$ will be discussed in $\mathsection$
\ref{subsec:inference}. The bandwidths can be chosen as estimated
optimal bandwidths by nonparametric smoothing methods, such that $h_{a}=O_{\bbP}\{n^{-1/(2q_{a}+d_{a})}\}$,
where $d_{a}=\dim(\mathcal{S}_{\E\{Y(a)|\bfX\}})$ ($a=0,1$). \end{step}

\begin{step} Estimate the central mean subspace $\mathcal{S}_{\E(D|\bfX)}$
based on $\{(\widehat{D}_{i},X_{i}):$ $i=1,\dots,n\}$. Let $\tau(\bfu;B)=\E\{Y(1)-Y(0)\cond B^{\T}\bfX=\bfu\}$.
Given $B$, the kernel smoothing estimator of $\tau(\bfu;B)$ is 
\begin{align}
\widehat{\tau}(\bfu;B)=\frac{\sum_{j=1}^{n}\widehat{D}_{j}\mathcal{K}_{q,h}(B^{\T}\bfX_{j}-\bfu)}{\sum_{j=1}^{n}\mathcal{K}_{q,h}(B^{\T}\bfX_{j}-\bfu)}.\label{eq:tau-hat-B}
\end{align}
We then estimate $(d_{\mathrm{\tau}},B_{\mathrm{\tau}})$ and a suitable
bandwidth for $\widehat{\tau}(\bfu;B)$ by the minimizer $(\widehat{d},\widehat{B},\widehat{h})$
of the following criterion: 
\begin{align*}
\textsc{cv}(d,B,h)=n^{-1}\sum_{i=1}^{n}\{\widehat{D}_{i}-\widehat{\tau}^{-i}(B^{\T}\bfX_{i};B)\}^{2},
\end{align*}
where the superscript $-i$ indicates the estimator (\ref{eq:tau-hat-B})
based on data without the $i$th subject. Here, $q>\max(d/2+1,2)$
is also specified for each working dimension $d$. \end{step}

\begin{step}Estimate $\tau(\bfx)$ by $\widehat{\tau}(\widehat{B}^{\T}\bfx;\widehat{B})$
with some suitable choice of $(q_{\tau},h_{\tau})$, which will be
further discussed in $\mathsection$ \ref{subsec:inference}. \end{step}

\begin{remark}

In Step 2, we estimate the structural dimension and the basis matrix
simultaneously. On the other hand, \citet{LiangY2020Semiparametric}
considered the multiple index model with a fixed dimension of the
index and proposed the semiparametric efficient score of $B_{\tau}$.
As we will show in Theorem \ref{thm:B_a}, the asymptotic distribution
of $\widehat{B}$ does not affect the asymptotic distribution of the
estimated conditional average treatment effect as long as $\widehat{B}$
is root-$n$ consistent. Therefore, it is not necessary to pursue
the semiparametric efficiency estimation of the central mean subspace
in our context. \end{remark}

\begin{remark} An alternative method of imputing the counterfactual
outcomes is matching. To fix ideas, we consider matching without replacement
and with the number of matches fixed at one. Then the matching procedure
becomes nearest neighbor imputation \citep{LittleR2002Statistical}.
Without loss of generality, we use the Euclidean distance to determine
neighbors; the discussion applies to other distances \citep{AbadieI2006Large}.
Let $\mathcal{J}_{i}$ be the index set for the matched subject of
$i$th subject. Define the imputed missing outcome as $\widetilde{Y}_{i}(A_{i})=Y_{i}$
and $\widetilde{Y}_{i}(1-A_{i})=\sum_{j\in\mathcal{J}_{i}}Y_{j}$.
Then the individual causal effect can be estimated by $\widehat{D}_{\text{MAT},i}=\widetilde{Y}_{i}(1)-\widetilde{Y}_{i}(0)$.
Matching uses the full vector of confounders to determine the distance
and corresponding neighbors. When the number of confounders gets larger,
this distance may be too conservative to determine proper neighbors
due to the curse of dimensionality. In the simulation studies, we
find that the estimation of $\mathcal{S}_{\E(D|\bfX)}$ based on $\widehat{D}_{\text{MAT},i}$
has a poor performance.\end{remark}

\begin{remark}

Instead of imputing the counterfactual outcomes, weighting can also
be used to estimate $D_{i}$ directly. Several authors have considered
an adjusted outcome $\widehat{D}_{\text{IPW},i}=\{A_{i}-\pi(\bfX_{i})\}Y_{i}/[\pi(\bfX_{i})\{1-\pi(\bfX_{i})\}]$
by inverse propensity score weighting. The adjusted outcome is unbiased
of $\tau(X_{i})$ due to 
\begin{align*}
\E(\widehat{D}_{\text{IPW},i}\cond\bfX_{i})=\E\left\{ \frac{A_{i}Y_{i}}{\pi(\bfX_{i})}-\frac{(1-A_{i})Y_{i}}{1-\pi(\bfX_{i})}\cond\bfX_{i}\right\} =\E\{Y_{i}(1)-Y_{i}(0)\cond\bfX_{i}\}=\tau(\bfX_{i}).
\end{align*}
This approach is attractive in clinical trials, where $\pi(\bfX_{i})$
is known by trial design. In observational studies, $\pi(\bfX_{i})$
is usually unknown and needs to be estimated. \citet{AbrevayaHL2015Estimating}
considered kernel regression to estimate $\pi(\bfX_{i})$. To avoid
possible curse of dimensionality and keep the nonparametric merit,
we can perform a prior dimension reduction to find $B_{\pi}$ such
that $\pi(\bfX_{i})=\pr(A_{i}=1\cond B_{\pi}^{\T}\bfX_{i})$. Then
an improved estimator of $\pi(\bfX_{i})$ is 
\begin{align*}
\widehat{\pi}(\widehat{B}_{\pi}^{\T}\bfX_{i};\widehat{B}_{\pi})=\frac{\sum_{j=1}^{n}A_{j}\mathcal{K}_{q,h}(\widehat{B}_{\pi}^{\T}\bfX_{j}-\widehat{B}_{\pi}^{\T}\bfX_{i})}{\sum_{j=1}^{n}\mathcal{K}_{q,h}(\widehat{B}_{\pi}^{\T}\bfX_{j}-\widehat{B}_{\pi}^{\T}\bfX_{i})},
\end{align*}
where $\widehat{B}_{\pi}$ can be obtained similarly following Step
1 in $\mathsection$ \ref{subsec:Imputation-and-Estimation} by changing
the outcome to $A$. However, the estimator $\widehat{D}_{\text{IPW},i}=\{A_{i}-\widehat{\pi}(\widehat{B}_{\pi}^{\T}\bfX_{i};\widehat{B}_{\pi})\}Y_{i}/[\widehat{\pi}(\widehat{B}_{\pi}^{\T}\bfX_{i};\widehat{B}_{\pi})\{1-\widehat{\pi}(\widehat{B}_{\pi}^{\T}\bfX_{i};\widehat{B}_{\pi})\}]$
still suffers from the instability due to the inverse weighting, especially
when $\widehat{\pi}(\widehat{B}_{\pi}^{\T}\bfX_{i};\widehat{B}_{\pi})$
is close to zero or one. It is well known that the augmented inverse
propensity weighted estimator reduces this instability by combining
inverse propensity weighting and outcome regressions. Specifically,
the corresponding estimator of $D_{i}$ is 
\begin{align*}
\widehat{D}_{\text{AIPW},i}=\frac{Y_{i}-\{1-\widehat{\pi}(\widehat{B}_{\pi}^{\T}\bfX_{i};\widehat{B}_{\pi})\}\widehat{\mu}_{1}(\widehat{B}_{1}^{\T}\bfX_{i};\widehat{B}_{1})-\widehat{\pi}(\widehat{B}_{\pi}^{\T}\bfX_{i};\widehat{B}_{\pi})\widehat{\mu}_{0}(\widehat{B}_{0}^{\T}\bfX_{i};\widehat{B}_{0})}{\widehat{\pi}(\widehat{B}_{\pi}^{\T}\bfX_{i};\widehat{B}_{\pi})\{1-\widehat{\pi}(\widehat{B}_{\pi}^{\T}\bfX_{i};\widehat{B}_{\pi})\}}.
\end{align*}
One can easily show that $\E(\widehat{D}_{\text{AIPW},i}\cond\bfX_{i})$
is asymptotically unbiased of $\tau(\bfX_{i})$. The estimator $\widehat{D}_{\text{AIPW},i}$
is a refined version of \citet{LeeOW2017Doubly}, in which the propensity
scores are estimated without a prior dimension reduction. Our simulation
shows that the estimated central mean subspace and conditional average
treatment effect based on $\widehat{D}_{i}$ and $\widehat{D}_{\text{AIPW},i}$
are comparable and both outperform those based on $\widehat{D}_{\text{MAT},i}$
and $\widehat{D}_{\text{IPW},i}$. Since $\widehat{D}_{\text{AIPW},i}$
requires an extra dimension reduction on $\pi(\bfx)$ and, hence,
more computational time, our proposed $\widehat{D}_{i}$ is more computationally
efficient in practice. \end{remark}

\subsection{Inference\label{subsec:inference}}

In this subsection, we derive the large sample properties of $\widehat{B}$
and $\widehat{\tau}(\widehat{B}^{\T}\bfx;\widehat{B})$. Based on
the notation and regularity conditions in the Supplementary Material,
we first establish the following theorem for the prior sufficient
dimension reduction for $\mu_{a}(\bfx)$ $(a=0,1)$.

\begin{theorem}\label{thm:B_a} Suppose that Assumption \ref{assump:ignorable}
and Conditions A1--A5 are satisfied. Then $\pr(\widehat{d}_{a}=d_{a})\to1$,
$\widehat{h}_{a}=O_{\bbP}\{n^{-1/(2q+d_{a})}\}$, and 
\begin{align*}
n^{1/2}\vctl(\widehat{B}_{a}-B_{a})1(\widehat{d}_{a}=d_{a})=n^{1/2}\sum_{i=1}^{n}\xi_{B_{a},i}+o_{\bbP}(1)\to{\rm N}(0,\Sigma_{B_{a}})
\end{align*}
in distribution as $n\to\infty$, where $\xi_{B_{a}}=\{V_{a}(B_{a})\}^{-1}S_{a}(B_{a})$
and $\Sigma_{B_{a}}=\{V_{a}(B_{a})\}^{-1}\E\{S_{a}^{\otimes2}(B_{a})\}\{V_{a}(B_{a})\}^{-1}$
($a=0,1$). \end{theorem}

Exact forms of $V_{a}(B_{a})$ and $S_{a}(B_{a})$ are presented in
the Supplementary Material. Theorem \ref{thm:B_a} is a modification
of results in \citet{HuangC2017Effective} and hence we omit the proof.
Theorem \ref{thm:B_a} is a building block to derive the asymptotic
distributions of the estimated central mean space and the proposed
estimator for $\tau(\bfx)$, taking into account the fact that $D_{i}$
is imputed.

\begin{theorem}\label{thm:B_M} Suppose that Assumption \ref{assump:ignorable}
and Conditions A1--A8 are satisfied. Then $\pr(\widehat{d}=d_{{\tau}})\to1$,
$\widehat{h}=O_{\bbP}\{n^{-1/(2q+d_{{\tau}})}\}$, and 
\begin{align*}
n^{1/2}\vctl(\widehat{B}-B_{{\tau}})1(\widehat{d}=d_{{\tau}})=n^{1/2}\sum_{i=1}^{n}\xi_{B_{{\tau}},i}+o_{\bbP}(1)\to{\rm N}(0,\Sigma_{B_{{\tau}}})
\end{align*}
in distribution as $n\to\infty$, where $\xi_{B_{{\tau}}}=\{V(B_{{\tau}})\}^{-1}S(B_{{\tau}})$
and $\Sigma_{B_{{\tau}}}=\{V(B_{{\tau}})\}^{-1}\E\{S^{\otimes2}(B_{{\tau}})\}\{V(B_{{\tau}})\}^{-1}$.
\end{theorem}

Exact forms of $V(B_{\tau})$ and $S(B_{\tau})$ are presented in
the Supplementary Material.

\begin{theorem}\label{thm:tau} Suppose that Assumption \ref{assump:ignorable}
and Conditions A1--A10 are satisfied. Then, 
\begin{align*}
(nh_{\tau}^{d_{{\tau}}})^{1/2}\{\widehat{\tau}(\widehat{B}^{\T}\bfx;\widehat{B})-\tau(\bfx)-h_{\tau}^{q_{\tau}}\gamma(\bfx)\}\to{\rm N}\{0,\sigma_{\tau}^{2}(\bfx)\}
\end{align*}
in distribution as $n\to\infty$, where 
\begin{align*}
 & \gamma(\bfx)=\left.\frac{\kappa\partial_{\bfu}^{q_{\tau}}\{\E(Z\cond B_{{\tau}}^{\T}\bfX=\bfu)f_{B_{{\tau}}^{\T}\bfX}(\bfu)\}-\E(Z\cond B_{{\tau}}^{\T}\bfX=\bfu)\partial_{\bfu}^{q_{\tau}}f_{B_{{\tau}}^{\T}\bfX}(\bfu)}{f_{B_{{\tau}}^{\T}\bfX}(\bfu)}\right\vert _{\bfu=B_{{\tau}}^{\T}\bfx},\\
 & \sigma_{\tau}^{2}(\bfx)=\frac{\left\{ \int K_{q_{\tau}}^{2}(s){\rm d}s\right\} ^{d_{\tau}}\var[Z+\{1-\pi(\bfX)\}\varepsilon_{1}-\pi(\bfX)\varepsilon_{0}\cond B_{{\tau}}^{\T}\bfX=B_{{\tau}}^{\T}\bfx]}{f_{B_{{\tau}}^{\T}\bfX}(B_{{\tau}}^{\T}\bfx)},
\end{align*}
$\kappa={\int s^{q_{\tau}}K_{q_{\tau}}(s){\rm d}s}/{q_{\tau}!}$,
$Z=(2A-1)\{Y-\mu_{1-A}(B_{1-A}^{\T}\bfX;B_{1-A})\}$, and $\varepsilon_{a}=\{Y-\mu_{a}(\bfX)\}1(A=a)$
($a=0,1$). \end{theorem}

The proofs of Theorems \ref{thm:B_M}--\ref{thm:tau} are given in
the Supplementary Material. The proof of Theorem \ref{thm:B_M} is
similar to that of Theorem \ref{thm:B_a}. The main difference is
that the outcome contributing to the asymptotic distribution is now
$Z$ instead of the counterfactual $D$. The proof of Theorem \ref{thm:tau}
mainly focuses on approximating the influence function coupled with
the difference between imputed and non-imputed counterfactual outcomes.

\begin{remark} One should note that the asymptotic bias of $\widehat{\mu}_{a}(\bfu;B)$
is not involved in the asymptotic distribution of $\widehat{\tau}(\widehat{B}^{\T}\bfx;\widehat{B})$.
This is an important result of Condition A6, which ensures that the
convergence rate of $\widehat{\mu}_{a}(\bfu;B)-\mu_{a}(\bfu;B)$ is
always faster then that of $\widehat{\tau}(\bfu;B)-\E(Z\cond B^{\T}\bfX=\bfu)$.
\end{remark}

\begin{remark} \label{rmk1}The most important feature of Theorem
\ref{thm:tau} is that the asymptotic variance of $\widehat{B}$ is
not involved in the asymptotic variance of $\widehat{\tau}(\widehat{B}^{\T}\bfx;\widehat{B})$.
More precisely speaking, $\widehat{\tau}(\widehat{B}^{\T}\bfx;\widehat{B})$
has the same asymptotic variance as the one of $\widehat{\tau}({B}_{{\tau}}^{\T}\bfx;{B}_{{\tau}})$.
The reason is that $\Vert\widehat{B}-B_{{\tau}}\Vert=O_{\bbP}(n^{-1/2})$,
which is much faster then the convergence rate $O_{\bbP}[h_{\tau}^{q_{\tau}}+\{\log n/(nh_{\tau}^{d_{{\tau}}})\}^{1/2}]$
of $\widehat{\tau}({B}_{{\tau}}^{\T}\bfx;{B}_{{\tau}})-\tau(\bfx)$.
\end{remark}

Based on Theorem \ref{thm:tau}, we can make inference of $\tau(\bfx)$
by estimating the asymptotic bias and variance. However, in practice,
direct estimates of $\gamma(\bfx)$ and $\sigma_{\tau}^{2}(\bfx)$
are usually unstable, especially when the imputed counterfactual outcomes
are involved. For a pre-specified $q_{\tau}$ that satisfies Condition
A10, we propose a under-smooth strategy such that the asymptotic bias
is dominated by the asymptotic variance. We propose to choose an optimal
bandwidth $h_{\tau,\mathrm{opt}}=O\{n^{-1/(2q_{\tau}+d_{{\tau}})}\}$
by using standard cross-validation criterion and use $h_{\tau}=h_{\tau,\mathrm{opt}}n^{-\delta_{\tau}}$
for some small positive value $\delta_{\tau}$ in the inference procedure.
We then use a bootstrapping method to estimate the asymptotic distribution
of $\widehat{\tau}(\widehat{B}^{\T}\bfx;\widehat{B})-\tau(\bfx)$.

Let $\xi_{i}$ ($i=1,\dots,n$) be independent and identically distributed
from a certain distribution with mean $\mu_{\xi}$ and variance $\sigma_{\xi}^{2}$.
Then $w_{i}=\xi_{i}/\sum_{j=1}^{n}\xi_{j}$ ($i=1,\dots,n$) are exchangeable
random weights. The bootstrapped estimator $\widehat{\tau}^{\ast}(\bfx)$
is calculated as 
\begin{align*}
\widehat{\tau}^{\ast}(\bfx)=\frac{\sum_{j=1}^{n}w_{j}\widehat{D}_{j}^{\ast}\mathcal{K}_{q_{\tau},h_{\tau}}(\widehat{B}^{\T}\bfX_{j}-\widehat{B}^{\T}\bfx)}{\sum_{j=1}^{n}w_{j}\mathcal{K}_{q_{\tau},h_{\tau}}(\widehat{B}^{\T}\bfX_{j}-\widehat{B}^{\T}\bfx)},
\end{align*}
where 
\begin{align*}
 & \widehat{D}_{i}^{\ast}=A_{i}\{Y_{i}-\widehat{\mu}_{0}^{\ast}(\widehat{B}_{0}^{\T}\bfX_{i};\widehat{B}_{0})\}+(1-A_{i})\{\widehat{\mu}_{1}^{\ast}(\widehat{B}_{1}^{\T}\bfX_{i};\widehat{B}_{1})-Y_{i}\},\\
 & \widehat{\mu}_{a}^{\ast}(\bfu;B)=\frac{\sum_{j=1}^{n}w_{j}Y_{j}1(A_{j}=a)\mathcal{K}_{q_{a},h_{a}}(\widehat{B}_{a}^{\T}\bfX_{j}-\bfu)}{\sum_{j=1}^{n}w_{j}1(A_{j}=a)\mathcal{K}_{q_{a},h_{a}}(\widehat{B}_{a}^{\T}\bfX_{j}-\bfu)}\enspace(a=0,1).
\end{align*}
According to Remark \ref{rmk1}, $\widehat{B}$, $\widehat{B}_{0}^{\T}$,
and $\widehat{B}_{1}^{\T}$ require no bootstrapping in the inference,
which highly reduces the computational burden in practice.

The asymptotic variance of $\widehat{\tau}(\widehat{B}^{\T}\bfx;\widehat{B})$
is estimated by $[\mathrm{se}\{\widehat{\tau}^{\ast}(\bfx)\}\mu_{\xi}/\sigma_{\xi}]^{2}$,
where $\mathrm{se}(\cdot)$ denote the standard error of $N$ bootstrapped
estimators. The confidence region of $\tau(\bfx)$ with $1-\alpha$
confidence level can then be constructed as 
\begin{align*}
\widehat{\tau}(\widehat{B}^{\T}\bfx;\widehat{B})\pm\mathcal{Z}_{1-\alpha/2}\mathrm{se}\{\widehat{\tau}^{\ast}(\bfx)\}\frac{\mu_{\xi}}{\sigma_{\xi}},
\end{align*}
where $\mathcal{Z}_{p}$ is the $p$th quantile of the standard normal
distribution.

\section{Simulation\label{sec:simulation}}

\subsection{Data generating processes}

In this section we present a Monte Carlo exercise aimed at evaluating
the finite-sample accuracy of the asymptotic approximations given
in the previous section. 
The covariates $\bfX=(X_{1},\dots,X_{10})$ are generated from independent
and identical $\mathrm{Unif}(-3^{1/2},3^{1/2})$. The propensity score
is $\mathrm{logit}\{\pi(\bfX)\}=0.5(1+X_{1}+X_{2}+X_{3})$. The percentage
of treated is about $60\%$. The potential outcomes are designed as
following two settings: 
\begin{itemize}
\item[M1.] $Y(0)=X_{1}-X_{2}+\varepsilon(0)$ and $Y(1)=2X_{1}+X_{3}+\varepsilon(1)$,
where $\varepsilon(0)$ and $\varepsilon(1)$ independently follow
$N(0,0.02^{2})$. Hence, the conditional average treatment effect
is $\tau(\bfx)=x_{1}+x_{2}+x_{3}$, and the central mean subspace
is $\csp\{(1,1,1,0,\dots,0)^{\T}\}$. 
\item[M2.] $Y(0)=(X_{1}+X_{3})(X_{2}-1)+\varepsilon(0)$ and $Y(1)=2X_{2}(X_{1}+X_{3})+\varepsilon(1)$,
where $\varepsilon(0)$ and $\varepsilon(1)$ independently follow
$N(0,0.02^{2})$. Hence, the conditional average treatment effect
is $\tau(\bfx)=(x_{1}+x_{3})^{2}(x_{2}+1)^{2}$, and the central mean
subspace is $\csp\{(1,0,1,0,\dots,0)^{\T},(0,1,0,\dots,0)^{\T}\}$. 
\end{itemize}
The sample size ranges from $n=250$ and $n=500$. All the results
are based on 1000 replications.

\subsection{Competing estimators and simulation results}

First, we compare the finite-sample performance of the estimated central
mean subspaces using different imputed or adjusted outcomes. In addition
to our proposed $\widehat{D}_{i}$, the nearest neighbor imputation
$\widehat{D}_{\text{MAT},i}$, the inverse weighted outcome $\widehat{D}_{\text{IPW},i}$
as well as $\widehat{D}_{\text{AIPW},i}$, we also consider $\widehat{D}_{\bfX,i}=(2A_{i}-1)\{Y_{i}-\widehat{\mu}_{1-A_{i}}(\bfX_{i};I_{p})\}$,
which is the imputed outcome without any dimension reduction. To compare
the information loss for counterfactual outcomes and prior dimension
reduction, we further perform the dimension reduction based on the
true individual effect $D_{i}$ and the imputed outcome $\widehat{D}_{\text{OR},i}=(2A_{i}-1)\{Y_{i}-\widehat{\mu}_{1-A_{i}}(\bfX_{i};B_{1-A_{i}})\}$
based on true oracle central mean subspaces of the prognostic scores.
The proportions of estimated structural dimension and the mean squared
errors $\Vert\widehat{B}(\widehat{B}^{\T}\widehat{B})^{-1}\widehat{B}^{T}-B_{\tau}(B_{\tau}^{\T}B_{\tau})^{-1}B_{\tau}^{\T}\Vert^{2}$
of the estimated central mean subspaces are displayed in Table \ref{tab:simulation}.
In general, all the proportions of selecting correct structural dimension
tend to one and the mean squared errors tend to zero as sample size
increases. Moreover, our proposed estimator outperforms the others
and is comparable with respect to the simulated estimators based on
$\widehat{D}_{\text{OR},i}$.

Second, we compare the finite-sample performance of the estimated
conditional average treatment effects, which include our proposed
estimator $\widehat{\tau}(\widehat{B}^{\T}\bfx;\widehat{B})$, the
estimator $\widehat{\tau}_{\bfX}(\bfx)$ based on imputed outcome
$\widehat{D}_{\bfX,i}$, the estimator $\widehat{\tau}_{\text{MAT}}(\bfx)$
based on the imputed outcome $\widehat{D}_{\text{MAT},i}$, the estimator
$\widehat{\tau}_{\text{IPW}}(\bfx)$ based on the adjusted outcome
$\widehat{D}_{\text{IPW},i}$, and the estimator $\widehat{\tau}_{\text{AIPW}}(\bfx)$
based on the adjusted outcome $\widehat{D}_{\text{AIPW},i}$. In addition,
we also estimate the conditional average treatment effect by using
the difference of two estimated prognostic scores $\widehat{\tau}_{\text{prog}}(\bfx)=\widehat{\mu}_{1}(\widehat{B}_{1}^{\T}\bfx;\widehat{B}_{1})-\widehat{\mu}_{0}(\widehat{B}_{0}^{\T}\bfx;\widehat{B}_{0})$.
The smoothing estimator $\widehat{\tau}_{0}(\bfx)$ based on $D_{i}$
is also considered as a reference to demonstrate the information loss.
The conditional average treatment effects are evaluated at $\bfx=(0,\dots,0)^{\T}$.
The means, standard deviations, and the mean squared errors are displayed
in Table \ref{tab:simulation_condTE}. In general, our proposed estimator
and the $\widehat{\tau}_{\text{AIPW}}$ have comparable performance,
and both of them outperform the others.

Finally, we construct confidence intervals and inference for the conditional
average treatment effects by using bootstrapping. Here naive bootstrapping
is adopted. That is, $(w_{1},\dots,w_{n})$ follows a multinomial
distribution with number of trials being $n$ and event probabilities
$(1/n,\dots,1/n)$. Table \ref{tab:inference} includes the standard
deviations, bootstrapped standard errors and 95\% quantile intervals
of estimated conditional average treatment effects, as well as the
normal-type 95\% confidence intervals with corresponding coverage
probabilities and quantile-type 95\% confidence intervals with corresponding
coverage probabilities for true conditional average treatment effect.
As expected, the standard errors get close to the standard deviations,
and the coverage probabilities tend to the nominal level when the
sample size gets larger.

\begin{table}
\caption{\label{tab:simulation}The proportions of $\widehat{d}$ and the mean
squared errors (MSE) of $\widehat{B}$ under different model settings,
sample sizes ($n$), and imputation of $D_{i}$}

\vspace{0.25cm}

\centering %
\begin{tabular}{cclcccccc}
\toprule 
 &  &  & \multicolumn{5}{c}{proportions of $\widehat{d}$} & \tabularnewline
\cmidrule{4-8} \cmidrule{5-8} \cmidrule{6-8} \cmidrule{7-8} \cmidrule{8-8} 
model & $n$ &  & 0 & 1 & 2 & 3 & $\geq$4 & MSE\tabularnewline
\midrule
\multirow{14}{*}{M1} & 250 & $\widehat{D}_{i}$ & 0.000 & 0.976 & 0.024 & 0.000 & 0.000 & 0.0293\tabularnewline
 &  & $\widehat{D}_{\bfX,i}$ & 0.000 & 0.716 & 0.246 & 0.037 & 0.001 & 0.5840\tabularnewline
 &  & $\widehat{D}_{\text{MAT},i}$ & 0.000 & 0.833 & 0.148 & 0.018 & 0.001 & 0.2927\tabularnewline
 &  & $\widehat{D}_{\text{IPW},i}$ & 0.000 & 0.680 & 0.229 & 0.087 & 0.004 & 0.7143\tabularnewline
 &  & $\widehat{D}_{\text{AIPW},i}$ & 0.000 & 0.955 & 0.045 & 0.000 & 0.000 & 0.0555\tabularnewline
 &  & $D_{i}$ & 0.000 & 0.999 & 0.001 & 0.000 & 0.000 & 0.0013\tabularnewline
 &  & $\widehat{D}_{\text{OR},i}$ & 0.000 & 0.979 & 0.021 & 0.000 & 0.000 & 0.0267\tabularnewline
\cmidrule{2-9} \cmidrule{3-9} \cmidrule{4-9} \cmidrule{5-9} \cmidrule{6-9} \cmidrule{7-9} \cmidrule{8-9} \cmidrule{9-9} 
 & 500 & $\widehat{D}_{i}$ & 0.000 & 0.985 & 0.015 & 0.000 & 0.00 & 0.0171\tabularnewline
 &  & $\widehat{D}_{\bfX,i}$ & 0.000 & 0.676 & 0.295 & 0.029 & 0.00 & 0.5392\tabularnewline
 &  & $\widehat{D}_{\text{MAT},i}$ & 0.000 & 0.897 & 0.097 & 0.006 & 0.00 & 0.1588\tabularnewline
 &  & $\widehat{D}_{\text{IPW},i}$ & 0.000 & 0.615 & 0.256 & 0.119 & 0.01 & 0.6744\tabularnewline
 &  & $\widehat{D}_{\text{AIPW},i}$ & 0.000 & 0.980 & 0.020 & 0.000 & 0.00 & 0.0236\tabularnewline
 &  & $D_{i}$ & 0.000 & 0.999 & 0.001 & 0.000 & 0.00 & 0.0012\tabularnewline
 &  & $\widehat{D}_{\text{OR},i}$ & 0.000 & 0.985 & 0.015 & 0.000 & 0.00 & 0.0171\tabularnewline
\midrule
\multirow{14}{*}{M2} & 250 & $\widehat{D}_{i}$ & 0.000 & 0.000 & 0.995 & 0.005 & 0.000 & 0.0237\tabularnewline
 &  & $\widehat{D}_{\bfX,i}$ & 0.000 & 0.062 & 0.883 & 0.053 & 0.002 & 0.3222\tabularnewline
 &  & $\widehat{D}_{\text{MAT},i}$ & 0.000 & 0.050 & 0.894 & 0.052 & 0.004 & 0.3608\tabularnewline
 &  & $\widehat{D}_{\text{IPW},i}$ & 0.000 & 0.269 & 0.610 & 0.110 & 0.011 & 0.9581\tabularnewline
 &  & $\widehat{D}_{\text{AIPW},i}$ & 0.000 & 0.008 & 0.978 & 0.014 & 0.000 & 0.0616\tabularnewline
 &  & $D_{i}$ & 0.000 & 0.000 & 0.995 & 0.005 & 0.000 & 0.0119\tabularnewline
 &  & $\widehat{D}_{\text{OR},i}$ & 0.000 & 0.003 & 0.992 & 0.004 & 0.001 & 0.0243\tabularnewline
\cmidrule{2-9} \cmidrule{3-9} \cmidrule{4-9} \cmidrule{5-9} \cmidrule{6-9} \cmidrule{7-9} \cmidrule{8-9} \cmidrule{9-9} 
 & 500 & $\widehat{D}_{i}$ & 0.000 & 0.000 & 0.997 & 0.003 & 0.000 & 0.0139\tabularnewline
 &  & $\widehat{D}_{\bfX,i}$ & 0.000 & 0.008 & 0.955 & 0.035 & 0.002 & 0.1858\tabularnewline
 &  & $\widehat{D}_{\text{MAT},i}$ & 0.000 & 0.013 & 0.963 & 0.021 & 0.003 & 0.2040\tabularnewline
 &  & $\widehat{D}_{\text{IPW},i}$ & 0.000 & 0.165 & 0.714 & 0.109 & 0.012 & 0.7532\tabularnewline
 &  & $\widehat{D}_{\text{AIPW},i}$ & 0.000 & 0.001 & 0.995 & 0.004 & 0.000 & 0.0224\tabularnewline
 &  & $\widehat{D}_{\text{OR},i}$ & 0.000 & 0.000 & 1.000 & 0.000 & 0.000 & 0.0090\tabularnewline
 &  & $D_{i}$ & 0.000 & 0.000 & 1.000 & 0.000 & 0.000 & 0.0027\tabularnewline
\bottomrule
\end{tabular}
\end{table}

\begin{table}
\caption{\label{tab:simulation_condTE}The mean squared errors of estimated
conditional average treatment effects under different model settings
and sample sizes ($n$)}

\vspace{0.25cm}

\centering \resizebox{\textwidth}{!}{%
\begin{tabular}{cccccccccc}
\toprule 
model & $n$ &  & $\widehat{\tau}(\widehat{B}^{\T}\bfx;\widehat{B})$ & $\widehat{\tau}_{\bfX}(\bfx)$ & $\widehat{\tau}_{\text{MAT}}(\bfx)$ & $\widehat{\tau}_{\text{IPW}}(\bfx)$ & $\widehat{\tau}_{\text{AIPW}}(\bfx)$ & $\widehat{\tau}_{\text{prog}}(\bfx)$ & $\widehat{\tau}_{0}(\bfx)$\tabularnewline
\midrule
\multirow{6}{*}{M1} & 250 & mean & 0.003 & -0.025 & 0.094 & 0.008 & 0.002 & 0.003 & -0.000\tabularnewline
 &  & s.d. & 0.0493 & 0.2203 & 0.2325 & 0.5903 & 0.0532 & 0.0545 & 0.0258\tabularnewline
 &  & MSE & 0.0024 & 0.0492 & 0.0629 & 0.3485 & 0.0028 & 0.0030 & 0.0007\tabularnewline
\cmidrule{2-10} \cmidrule{3-10} \cmidrule{4-10} \cmidrule{5-10} \cmidrule{6-10} \cmidrule{7-10} \cmidrule{8-10} \cmidrule{9-10} \cmidrule{10-10} 
 & 500 & mean & -0.000 & 0.006 & 0.065 & -0.005 & -0.000 & 0.003 & -0.001\tabularnewline
 &  & s.d. & 0.0300 & 0.1474 & 0.1417 & 0.3642 & 0.0311 & 0.0310 & 0.0159\tabularnewline
 &  & MSE & 0.0009 & 0.0218 & 0.0243 & 0.1327 & 0.0010 & 0.0010 & 0.0003\tabularnewline
\midrule
\multirow{6}{*}{M2} & 250 & mean & -0.029 & -0.091 & -0.180 & -0.035 & -0.007 & -0.048 & 0.001\tabularnewline
 &  & s.d. & 0.1006 & 0.2072 & 0.3103 & 0.3803 & 0.1074 & 0.1399 & 0.0639\tabularnewline
 &  & MSE & 0.0110 & 0.0512 & 0.1288 & 0.1459 & 0.0116 & 0.0219 & 0.0041\tabularnewline
\cmidrule{2-10} \cmidrule{3-10} \cmidrule{4-10} \cmidrule{5-10} \cmidrule{6-10} \cmidrule{7-10} \cmidrule{8-10} \cmidrule{9-10} \cmidrule{10-10} 
 & 500 & mean & -0.015 & -0.104 & -0.157 & -0.010 & -0.002 & -0.024 & 0.001\tabularnewline
 &  & s.d. & 0.0651 & 0.1418 & 0.2024 & 0.2463 & 0.0566 & 0.0926 & 0.0410\tabularnewline
 &  & MSE & 0.0045 & 0.0309 & 0.0655 & 0.0607 & 0.0032 & 0.0092 & 0.0017\tabularnewline
\bottomrule
\end{tabular}}
\end{table}

\begin{table}
\caption{\label{tab:inference}The standard deviations (s.d.), bootstrapped
standard errors (s.e.), and 95\% quantile intervals (Q.I.) of estimated
conditional average treatment effects, and normal-type 95\% confidence
intervals (N.C.I.) with corresponding coverage probabilities (N.C.P.)
and quantile-type 95\% confidence intervals (Q.C.I.) with corresponding
coverage probabilities (Q.C.P.) for true conditional treatment treatment
effect}
\centering

\vspace{0.25cm}

\resizebox{\textwidth}{!}{%
\begin{tabular}{ccccccccc}
\toprule 
model & $n$ & s.d. & s.e. & Q.I. & N.C.I & N.C.P. & Q.C.I. & Q.C.P.\tabularnewline
\midrule
\multirow{2}{*}{M1} & 250 & 0.0493 & 0.0621 & (-0.095,0.107) & (-0.119,0.125) & 0.966 & (-0.119,0.124) & 0.975\tabularnewline
 & 500 & 0.0300 & 0.0365 & (-0.066,0.062) & (-0.072,0.071) & 0.965 & (-0.074,0.067) & 0.972\tabularnewline
\midrule
\multirow{2}{*}{M2} & 250 & 0.1006 & 0.0998 & (-0.226,0.159) & (-0.225,0.166) & 0.944 & (-0.224,0.167) & 0.921\tabularnewline
 & 500 & 0.0651 & 0.0645 & (-0.132,0.109) & (-0.142,0.111) & 0.951 & (-0.140,0.112) & 0.937\tabularnewline
\bottomrule
\end{tabular}}
\end{table}

\section{Empirical examples\label{sec:data}}

\subsection{The effect of maternal smoking on birth weight}

We apply our proposed method to two existing datasets to estimate
the effect of maternal smoking on birth weight conditional on different
levels of confounders. In the literature, many studies documented
that mother's health, educational and labor market status have important
effects on child birth weight \citep{CurrieA2011Human}. In particular,
maternal smoking is considered as the most important preventable negative
cause \citep{Kramer1987Intrauterine}. \citet{LeeOW2017Doubly} studied
the conditional average treatment effect of smoking given mother's
age. In this work, our goal is to fully characterize the conditional
average treatment effect of smoking on child birth weight given a
vector of important confounding variables while maintaining the interpretability.

\subsection{Pennsylvania data}

The first dataset consists of observations collected in 2002 from
mothers in Pennsylvania in the U.S.A. available from the STATA website
(http://www.stata-press.com/data/r13/cattaneo2.dta). Following \citet{LeeOW2017Doubly},
we focus on white and non-Hispanic mothers, leading to the sample
size 3754. The outcome $Y$ of interest is infant birth weight measured
in grams. The treatment variable $A$ is equal to 1 if the mother
is a smoker and 0 otherwise. The set of covariates $\bfX$ includes
the number of prenatal care visits, mother's educational attainment,
age, an indicator for the first baby, an indicator for alcohol consumption
during pregnancy, an indicator for the first prenatal visit in the
first trimester, and an indicator for whether there was a previous
birth where the newborn died. The continuous covariates are centralized
and standardized.

The estimated central mean subspace has dimension one. The coefficients
of estimated linear index and corresponding standard errors are displayed
in Table \ref{tab:data_index}. 
Figure \ref{fig:CATE_index} shows the estimated conditional average
treatment effect at different levels of linear index values along
with corresponding normal-type confidence intervals. In general, smoking
has significant effects on low birth weights, as detected in the existing
studies. In particular, this effect decreases when the linear index
value increases. Interestingly, the larger number of prenatal care
visits and the first baby lead to significantly smaller effects than
other confounding variables. This result shows that more frequent
prenatal care visits and whether it is a first pregnancy mitigate
the effect of smoking on low birth weights.

\subsection{North Carolina data}

The second dataset is based on the records between 1988 and 2002 by
the North Carolina Center Health Services. The dataset was analyzed
by \citet{AbrevayaHL2015Estimating} and can be downloaded from Prof.
Leili's website. To make a comparison with the Pennsylvania data,
we focus on white and first-time mothers and form a random sub-sample
with sample size $n=3754$ among the subjects collected in 2002. The
outcome $Y$ and the treatment variable $A$ remain the same as for
the Pennsylvania data. The set of covariates includes those used in
the analysis of Pennsylvania data but the indicator for the first
baby and the indicator for whether there was a previous birth where
the newborn died. Besides, it includes indicators for gestational
diabetes, hypertension, amniocentesis, and ultrasound exams.

The estimated central mean subspace has dimension one. The coefficients
of estimated linear index and corresponding standard errors are also
displayed in Table \ref{tab:data_index}. Figure \ref{fig:CATE_index}
shows the estimated conditional average treatment effect at different
levels of linear index values along with corresponding normal-type
confidence intervals. Similar to the results from the Pennsylvania
data, smoking has significant effects on low birth weights. Differently,
this effect decreases when the level of estimated linear index values
decreases. In particular, lower values of mothers educational attainment,
higher values of mothers age, the absence of hypertension, and the
amniocentesis significantly lead to larger treatment effects. This
result shows that mother's education attainment, age, and health status
are also important modifying factors of smoking on low birth weight.

\begin{table}
\caption{\label{tab:data_index}The estimated coefficients of linear indices
and corresponding standard errors (s.e.) for the Pennsylvania and
North Carolina data.}

\vspace{0.25cm}

\centering %
\begin{tabular}{cccrrrrr}
\toprule 
 &  &  & \multicolumn{2}{c}{Pennsylvania data} &  & \multicolumn{2}{c}{North Carolina data}\tabularnewline
\cmidrule{4-5} \cmidrule{5-5} \cmidrule{7-8} \cmidrule{8-8} 
\multicolumn{2}{c}{covariate } &  & coefficient  & s.e.  &  & coefficient  & s.e.\tabularnewline
\midrule
$X_{1}$  & prenatal visit number  &  & -0.668  & 0.0645  &  & 0.043  & 0.0719\tabularnewline
$X_{2}$  & education  &  & -0.059  & 0.2101  &  & -0.271  & 0.0477\tabularnewline
$X_{3}$  & age  &  & -0.210  & 0.3076  &  & 0.243  & 0.0485\tabularnewline
$X_{4}$  & first baby  &  & 1  &  &  &  & \tabularnewline
$X_{5}$  & alcohol  &  & 0.142  & 0.6103  &  & -0.101  & 0.2122\tabularnewline
$X_{6}$  & first prenatal visit  &  & 0.275  & 0.3224  &  & -0.104  & 0.1556\tabularnewline
$X_{7}$  & previous newborn death  &  & 0.169  & 0.1257  &  &  & \tabularnewline
$X_{8}$  & diabetes  &  &  &  &  & -0.129  & 0.1268\tabularnewline
$X_{9}$  & hypertension  &  &  &  &  & -0.333  & 0.1084\tabularnewline
$X_{10}$  & amniocentesis  &  &  &  &  & 1  & \tabularnewline
$X_{11}$  & ultrasound  &  &  &  &  & -0.006  & 0.1612\tabularnewline
\bottomrule
\end{tabular}
\end{table}

\begin{figure}
\caption{\label{fig:CATE_index}The estimated conditional average treatment
effects at different levels of linear index values with corresponding
confidence intervals.}

\vspace{0.25cm}

\centering \includegraphics[scale=0.4]{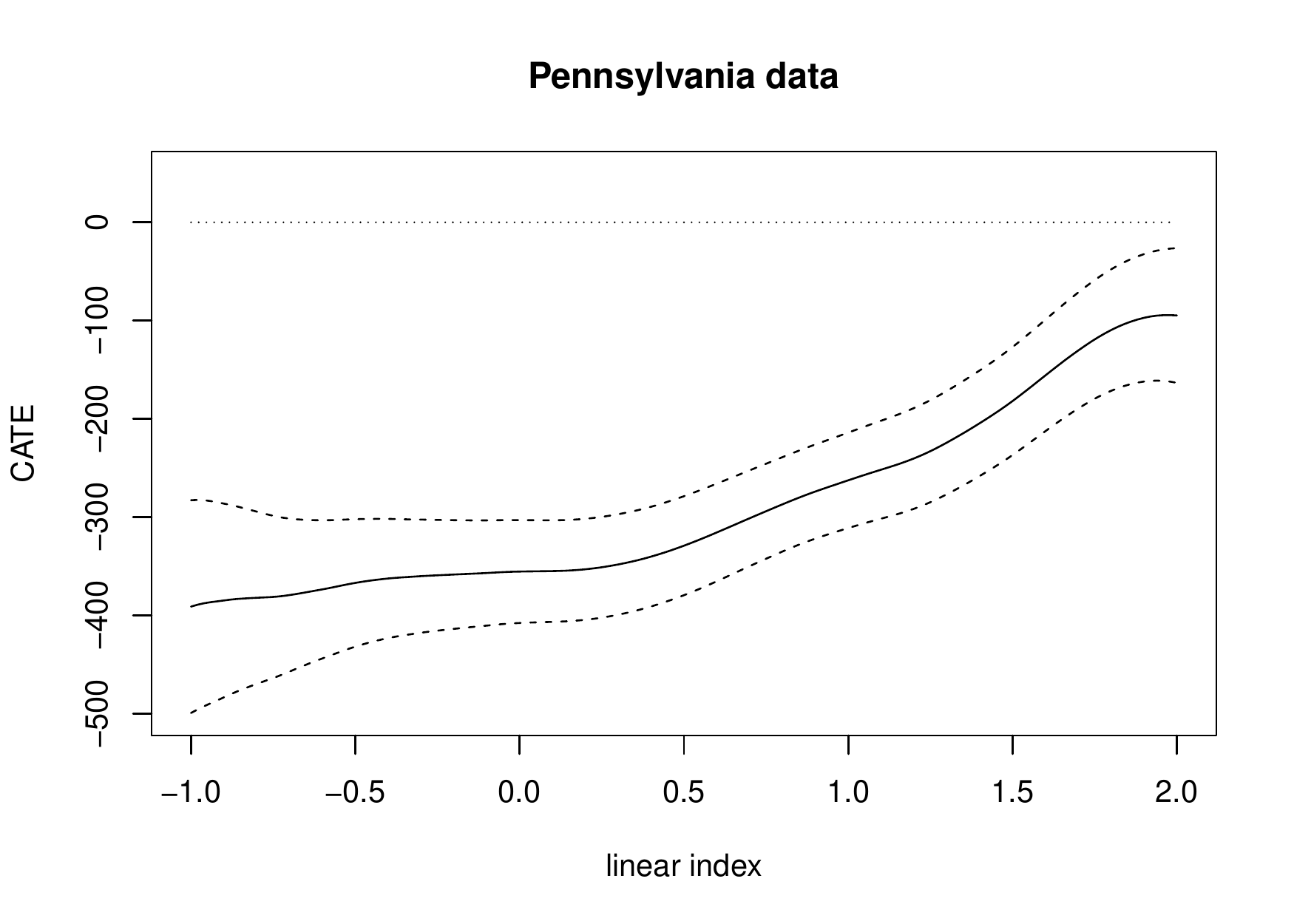} \includegraphics[scale=0.4]{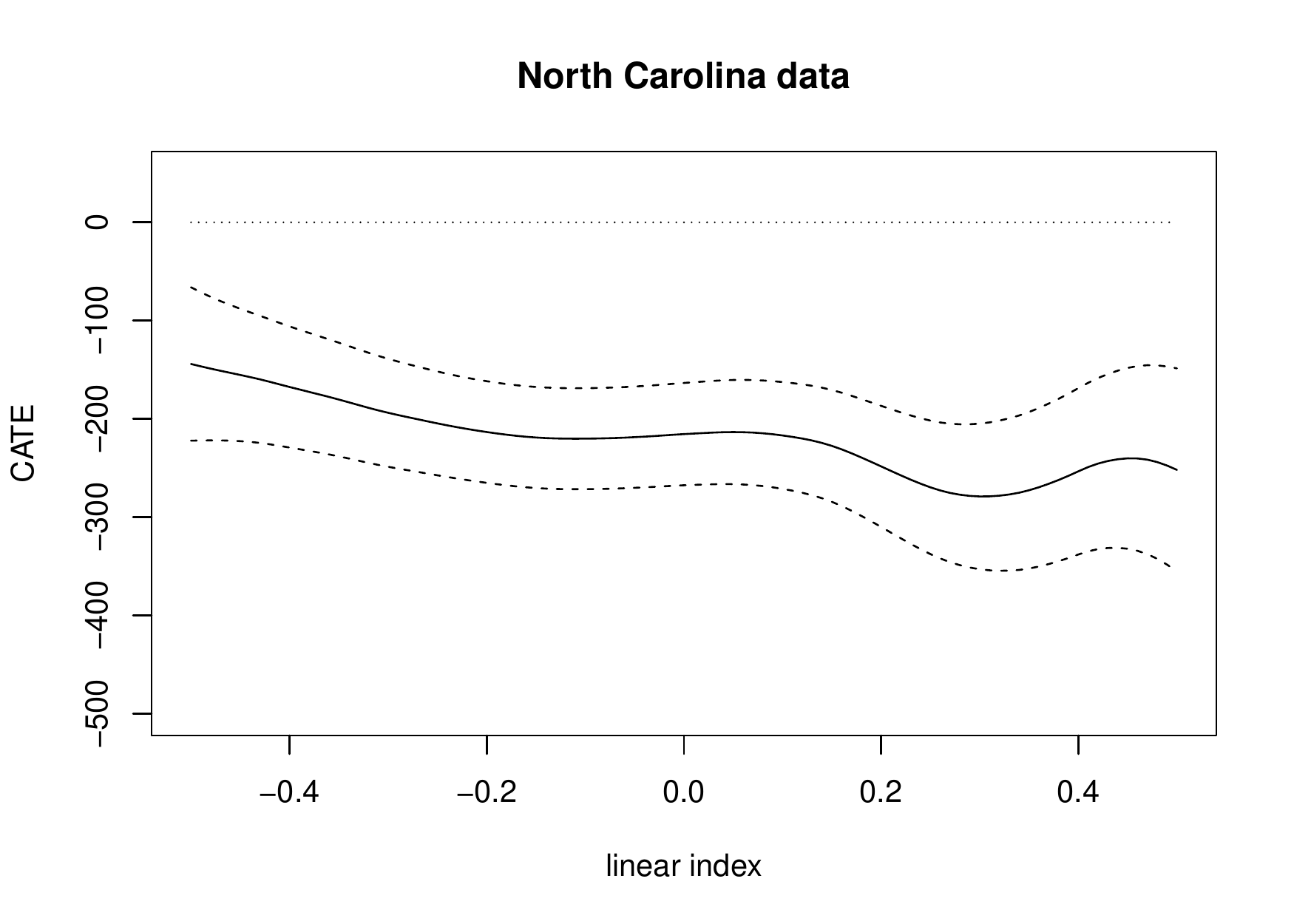} 
\end{figure}

\section{Discussion\label{sec:discuss}}


The proposed framework of robust inference of conditional average
treatment effect can be generalized in the following directions. First,
we use under-smoothing to avoid the asymptotic bias of the conditional
average treatment effect estimator. Without under-smoothing, the asymptotic
bias is not negligible but may be estimated empirically as in \citet{ChenC2019Nonparametric}.
We will investigate the possibility of a bias-corrected estimator
in the future. Second, we can extend to estimate the conditional average
treatment effect with continuous treatment. In this case, the first-stage
dimension reduction applies to the potential outcomes for a given
treatment level and a reference treatment level, and the second-stage
searches the central space for the contrast between the two prognostic
scores under the two levels. Third, the first-stage dimension reduction
is not confined to the central mean space but can be applied to a
transformation of the outcome $g\{Y(a)\}$ for any function $g(\cdot)$.
This allows the estimation of the general type of conditional treatment
effects such as conditional distribution or quantile treatment effects.
Similar to the main paper, we can also derive robust estimators for
these causal estimands.

\section*{Acknowledgment }

Dr. Huang is partially supported by MOST grant 108-2118-M-001-011-MY2.
Dr. Yang is partially supported by NSF grant DMS 1811245, NCI grant
P01 CA142538, NIA grant 1R01AG066883, and NIEHS grant 1R01ES031651.

\section*{Supplementary material\label{SM}}

Supplementary Material includes additional notation and the regularity
conditions and the proofs of Theorems \ref{thm:B_M}--\ref{thm:tau}.

\bibliographystyle{dcu}
\bibliography{MYH}

\pagebreak{}

\global\long\def\theequation{S\arabic{equation}}%
\setcounter{equation}{0}

\global\long\def\thelemma{S\arabic{lemma}}%
\setcounter{lemma}{0}

\global\long\def\thesection{S\arabic{section}}%
\setcounter{section}{0}

\global\long\def\thetheorem{S\arabic{theorem}}%
\setcounter{theorem}{0}

\global\long\def\thecondition{S\arabic{condition}}%
\setcounter{condition}{0}

\global\long\def\theremark{S\arabic{remark}}%
\setcounter{remark}{0}

\global\long\def\thestep{S\arabic{step}}%
\setcounter{step}{0}

\global\long\def\theassumption{S\arabic{assumption}}%
\setcounter{assumption}{0}

\global\long\def\theproposition{S\arabic{proposition}}%
\setcounter{equation}{0}

\global\long\def\thetable{S\arabic{table}}%
\setcounter{table}{0} 

\global\long\def\thepage{S\arabic{page}}%
\setcounter{page}{0} 
\begin{center}
\textbf{\huge{}{}{}{}{}{}{}{}{}{}Supplementary material for
``Robust inference of conditional average treatment effects using
dimension reduction'' }{\huge{}{}{}{}{}{}{}{}{}}\\
 {\huge{}{}{}{}{}{}{}{}{} }\textbf{\huge{}{}{}{}{}{}{}{}{}
\bigskip{}
 }\textbf{\large{}{}{}{}{}{}{}{}{}{}by Huang and Yang}{\large\par}
\par\end{center}

\begin{center}
 
\par\end{center}

\bigskip{}

\setcounter{page}{1}

$\mathsection$ \ref{sec:Additional-Notation-and} presents additional
notation and regularity conditions; $\mathsection$ \ref{sec:Preliminary-Lemmas}
establishes some preliminary lemmas; and $\mathsection$ \ref{sec:Proofs-of-Theorems}
provides the proofs of Theorems 2 and 3.

\section{Additional Notation and Regularity Conditions\label{sec:Additional-Notation-and}}

Let $(\cdot)^{\otimes}$ denote the Kronecker power of a vector and
let $\Vert\cdot\Vert$ represent the Frobenius norm of a matrix. Denote
$f_{B^{\T}\bfX}(\bfu)$ as the marginal density of $B^{\T}\bfX$,
\begin{align*}
 & f^{[m]}(\bfx,\bfu;B)=\partial_{\bfu}^{m}[\E\{(\bfX_{l}-\bfx_{l})^{\otimes m}\cond B^{\T}\bfX=\bfu\}f_{B^{\T}\bfX}(\bfu)],\\
 & E_{a}^{[m]}(\bfx,\bfu;B)=\partial_{\bfu}^{m}[\pr(A=a\cond B^{\T}\bfX=\bfu)\E\{(\bfX_{l}-\bfx_{l})^{\otimes m}\cond B^{\T}\bfX=\bfu\}f_{B^{\T}\bfX}(\bfu)],\\
 & F_{a}^{[m]}(\bfx,\bfu;B)=\partial_{\bfu}^{m}[\E\{Y1(A=a)\cond B^{\T}\bfX=\bfu\}\E\{(\bfX_{l}-\bfx_{l})^{\otimes m}\cond B^{\T}\bfX=\bfu\}f_{B^{\T}\bfX}(\bfu)],\\
 & G^{[m]}(\bfx,\bfu;B)=\partial_{\bfu}^{m}[\E(Z\cond B^{\T}\bfX=\bfu)\E\{(\bfX_{l}-\bfx_{l})^{\otimes m}\cond B^{\T}\bfX=\bfu\}f_{B^{\T}\bfX}(\bfu)],\enspace(a=0,1,\enspace m=0,1,2),
\end{align*}
where $Z=(2A-1)\{Y-\mu_{1-A}(B_{1-A}^{\T}\bfX;B_{1-A})\}$. We will
show that 
\[
\partial_{\vctl(B)}^{m}\widehat{\mu}_{a}(B^{\T}\bfx;B)\rightarrow\mu^{[m]}(\bfx;B)=\sum_{\ell=0}^{m}{m \choose \ell}F_{a}^{[\ell]}(\bfx,B^{\T}\bfx;B)E_{a,\mathrm{inv}}^{[m-\ell]}(\bfx,B^{\T}\bfx;B),
\]
and 
\[
\partial_{\vctl(B)}^{m}\widehat{\tau}(B^{\T}\bfx;B)\rightarrow\tau^{[m]}(\bfx;B)=\sum_{\ell=0}^{m}{m \choose \ell}G^{[\ell]}(\bfx,B^{\T}\bfx;B)f_{\mathrm{inv}}^{[m-\ell]}(\bfx,B^{\T}\bfx;B),
\]
uniformly as $n\rightarrow\infty,$ where 
\begin{align*}
 & f_{\mathrm{inv}}^{[0]}(\bfx,\bfu;B)={1}/{f_{B^{\T}\bfX}(\bfu)}, &  & E_{a,\mathrm{inv}}^{[0]}(\bfx,\bfu;B)={1}/{E_{a}^{[0]}(\bfx,\bfu;B)},\\
 & f_{\mathrm{inv}}^{[1]}(\bfx,\bfu;B)=-\frac{f^{[1]}(\bfx,\bfu;B)}{f_{B^{\T}\bfX}^{2}(\bfu)},\enspace &  & f_{\mathrm{inv}}^{[2]}(\bfx,\bfu;B)=\frac{2\{f^{[1]}(\bfx,\bfu;B)\}^{2}}{f_{B^{\T}\bfX}^{3}(\bfu)}-\frac{f^{[2]}(\bfx,\bfu;B)}{f_{B^{\T}\bfX}^{2}(\bfu)},\\
 & E_{a,\mathrm{inv}}^{[1]}(\bfx,\bfu;B)=-\frac{E_{a}^{[1]}(\bfx,\bfu;B)}{\{E_{a}^{[0]}(\bfx,\bfu;B)\}^{2}}, &  & \enspace E_{a,\mathrm{inv}}^{[2]}(\bfx,\bfu;B)=\frac{2\{E_{a}^{[1]}(\bfx,\bfu;B)\}^{2}}{E_{a}^{[0]}(\bfx,\bfu;B)}-\frac{E_{a}^{[2]}(\bfx,\bfu;B)}{\{E_{a}^{[0]}(\bfx,\bfu;B)\}^{2}}.
\end{align*}
According to the notation, we can define the corresponding score vectors
and information matrices of $\textsc{cv}_{a}(d,B,h)$ and $\textsc{cv}(d,B,h)$:
\begin{align*}
 & S_{a}(B)=-1(A=a)\{Y-\mu_{a}(B^{\T}\bfX;B)\}\mu^{[1]}(\bfX;B),\\
 & V_{a}(B)=\E(1(A=a)[\{\mu^{[1]}(\bfX;B)\}^{\otimes2}-\{Y-\mu_{a}(B^{\T}\bfX;B)\}\mu^{[2]}(\bfX;B)]),\\
 & S(B)=-\{Z-\E(Z\cond B^{\T}\bfX)\}\tau^{[1]}(\bfX;B),\\
 & V(B)=\E[\{\tau^{[1]}(\bfX;B)\}^{\otimes2}-\{Z-\E(Z\cond B^{\T}\bfX)\}\tau^{[2]}(\bfX;B)].
\end{align*}
In addition, let $B_{d,a}$ be the minimizer of $b_{a}^{2}(B)=\E[\{\mu_{a}(B^{\T}\bfX;B)-\mu(\bfX)\}^{2}]$
and let $B_{d,{\tau}}$ be the minimizer of $b_{{\tau}}^{2}(B)=\E[\{\E(Z\cond B^{\T}\bfX)-\tau(\bfX)\}^{2}]$
over all $p\times d$ matrices $B$. Then, $b_{a}^{2}(B)\to b_{a}^{2}(B_{d,a})$
implies $B\to B_{d,a}$ for $\csp(B)\nsupseteq\csp(B_{a})$, and $b_{{\tau}}^{2}(B)\to b_{{\tau}}^{2}(B_{d,{\tau}})$
implies $B\to B_{d,{\tau}}$ for $\csp(B)\nsupseteq\csp(B_{{\tau}})$.
The following regularity conditions are imposed for our theorems: 
\begin{itemize}
\item[A1] $\partial_{\bfu}^{q+m}\E\{(\bfX_{l}-\bfx_{l})^{\otimes m}\cond B^{\T}\bfX=\bfu\}$,
$\partial_{\bfu}^{q+2}f_{B^{\T}\bfX}(\bfu)$, $\partial_{\bfu}^{q+2}\pr(A=a\cond B^{\T}\bfX=\bfu)$,
$\partial_{\bfu}^{q+2}\E\{Y1{A=a}\cond B^{\T}\bfX=\bfu\}$, and $\partial_{\bfu}^{q+2}\E(Z\cond B^{\T}\bfX=\bfu)$
($a=0,1$, $m=1,2$), are Lipschitz continuous in $\bfu$ with the
Lipschitz constants being independent of $(\bfx,B)$. 
\item[A2] $\inf_{(\bfx,B)}f_{B^{\T}\bfX}(B^{\T}\bfx)>0$ and $\inf_{(\bfx,B)}\pr(A=a\cond B^{\T}\bfX=B^{\T}\bfx)>0$
($a=0,1$). 
\item[A3] For each working dimension $d>0$, $h$ falls in the interval $H_{\delta,n}=[h_{l}n^{-\delta},h_{u}n^{-\delta}]$
for some positive constants $h_{l}$ and $h_{u}$ and $\delta\in(1/(4q),1/\max\{2d+2,d+4\})$.
In particular, this requires $q>\max(d/2+1,2)$. 
\item[A4] $\inf_{\{B:d<d_{a}\}}b_{a}^{2}(B)>0$ and $b_{a}^{2}(B)=0$ if and
only if $B=B_{a}$ when $d=d_{a}$ ($a=0,1$). 
\item[A5] $V_{a}(B_{d,a})$ is non-singular for $d\geq d_{a}$ ($a=0,1$). 
\item[A6] For each working dimension $d$, $q_{a}>qd_{a}/d$ ($a=0,1$). 
\item[A7] $\inf_{\{B:d<d_{{\tau}}\}}b_{{\tau}}^{2}(B)>0$ and $b_{{\tau}}^{2}(B)=0$
if and only if $B=B_{{\tau}}$ when $d=d_{{\tau}}$. 
\item[A8] $V(B_{d,{\tau}})$ is non-singular for $d\geq d_{{\tau}}$. 
\item[A9] $h_{\tau}\to0$ and $nh_{\tau}^{d_{{\tau}}}\to\infty$. 
\item[A10] For each working dimension $d$, $q_{\tau}>qd_{\tau}/d$. 
\end{itemize}
Conditions A1--A2 are the smoothness and boundedness conditions for
the population functions to ensure the uniform convergence of kernel
estimators. Moreover, to remove the remainder terms in the approximation
of $\textsc{cv}(d,B,h)$ and $\textsc{cv}(d,B,h)$ to their target
functions, the constraints for the orders of kernel functions and
the bandwidths are drawn in Conditions A3 and A6. Conditions A4--A5
and A7--A8 ensure the identifiability of $B_{a}$ ($a=0,1$) and
$B_{{\tau}}$, respectively. The requirements of $h_{\tau}$ and $q_{\tau}$
used in $\widehat{\tau}(\widehat{B}^{\T}\bfx;\widehat{B})$ are given
in Condition A9--A10.

\section{Preliminary Lemmas\label{sec:Preliminary-Lemmas}}

The proofs of the main theorems rely on the following lemma:

\begin{lemma}\label{lem:tau_Dhat} Suppose that Assumption 1 and
Conditions A1--A6 are satisfied. Then, 
\begin{multline*}
\widehat{\tau}(\bfu;B)-\E(Z\cond B^{\T}\bfX=\bfu)\\
=\enspace\frac{1}{n}\sum_{i=1}^{n}[Z_{i}-\E(Z\cond B^{\T}\bfX=\bfu)+\{1-\pi(\bfX_{i})\}\varepsilon_{1,i}-\pi(\bfX_{i})\varepsilon_{0,i}]\omega_{h,i}(\bfu;B)+r_{n}(\bfu;B),
\end{multline*}
where $\varepsilon_{a,i}=\{Y_{i}-\mu_{a}(\bfX_{i})\}1(A_{i}=a),\ (a=0,1),$
$\omega_{h,i}(\bfu;B)=\mathcal{K}_{q,h}(B^{\T}\bfX_{i}-\bfu)/\sum_{j=1}^{n}\mathcal{K}_{q,h}(B^{\T}\bfX_{j}-\bfu),$
and $\sup_{(\bfu,B)}\vert r_{n}(\bfu,B)\vert=o_{\bbP}[h^{q}+\{\log n/(nh^{d})\}^{1/2}].$
\end{lemma} 
\begin{proof}
First note that 
\begin{multline*}
\widehat{\tau}(\bfu;B)-\E(Z\cond B^{\T}\bfX=\bfu)=\frac{1}{n}\{\widehat{D}_{i}-\E(Z\cond B^{\T}\bfX=\bfu)\}\omega_{h,i}(\bfu;B)\\
=\frac{1}{n}\{Z_{i}-\E(Z\cond B^{\T}\bfX=\bfu)\}\omega_{h,i}(\bfu;B)+\frac{1}{n}\{\widehat{D}_{i}-Z_{i}\}\omega_{h,i}(\bfu;B).
\end{multline*}
Further, 
\begin{align}
 & \frac{1}{n}(\widehat{D}_{i}-Z_{i})\omega_{h,i}(\bfu;B)\nonumber \\
=\enspace & \frac{1}{n}\sum_{i=1}^{n}[(1-A_{i})\{\widehat{\mu}_{1}(\widehat{B}_{1}^{\T}\bfX_{i};\widehat{B}_{1})-\mu_{1}(\bfX_{i})\}-A_{i}\{\widehat{\mu}_{0}(\widehat{B}_{0}^{\T}\bfX_{i};\widehat{B}_{0})-\mu_{0}(\bfX_{i})\}]\omega_{h,i}(\bfu;B)\nonumber \\
=\enspace & \frac{1}{n}\sum_{i=1}^{n}(1-A_{i})\{\widehat{\mu}_{1}({B}_{1}^{\T}\bfX_{i};{B}_{1})-\mu_{1}(\bfX_{i})\}\omega_{h,i}(\bfu;B)\nonumber \\
 & -\frac{1}{n}\sum_{i=1}^{n}A_{i}\{\widehat{\mu}_{0}({B}_{0}^{\T}\bfX_{i};{B}_{0})-\mu_{0}(\bfX_{i})\}\omega_{h,i}(\bfu;B)+O_{\bbP}(n^{-1/2})\nonumber \\
\stackrel{\triangle}{=}\enspace & I_{1}+I_{2}+O_{\bbP}(n^{-1/2}),\label{eq:DhatZ}
\end{align}
because of $\Vert\vctl(\widehat{B}_{a}-B_{a})\Vert=O_{\bbP}(n^{-1/2})$
by Theorem 1. Now let $\kappa_{a,h,i}(\bfu)=\mathcal{K}_{q_{a},h}(B_{a}^{\T}\bfX_{i}-\bfu)/\sum_{j=1}^{n}1(A_{j}=a)\mathcal{K}_{q_{a},h}(B_{a}^{\T}\bfX_{j}-\bfu)$.
Then, we decompose $I_{1}$ into 
\begin{align}
 & \frac{1}{n}\sum_{i=1}^{n}(1-A_{i})\widehat{\mu}_{1}({B}_{1}^{\T}\bfX_{i};{B}_{1})-\mu_{1}(\bfX_{i})\}\omega_{h,i}(\bfu;B)\nonumber \\
=\enspace & \frac{1}{n}\sum_{i=1}^{n}\{1-\pi(\bfX_{i})\}\omega_{h,i}(\bfu;B)\sum_{j=1}^{n}\{Y_{j}-\mu_{1}(\bfX_{i})\}1(A_{j}=1)\kappa_{1,h_{1},j}(B_{1}^{\T}\bfX_{i})\nonumber \\
 & +\frac{1}{n}\sum_{i=1}^{n}\{\pi(\bfX_{i})-A_{i}\}\omega_{h,i}(\bfu;B)\sum_{j=1}^{n}\{Y_{j}-\mu_{1}(\bfX_{i})\}1(A_{j}=1)\kappa_{1,h_{1},j}(B_{1}^{\T}\bfX_{i})\nonumber \\
=\enspace & \frac{1}{n}\sum_{i=1}^{n}\{1-\pi(\bfX_{i})\}\varepsilon_{1,i}\omega_{h,i}(\bfu;B)\nonumber \\
 & +\frac{1}{n}\sum_{i=1}^{n}\{1-\pi(\bfX_{i})\}\left\{ \sum_{j=1}^{n}\varepsilon_{1,j}\kappa_{1,h_{1},j}(B_{1}^{\T}\bfX_{i})-\varepsilon_{1,i}\right\} \omega_{h,i}(\bfu;B)\nonumber \\
 & +\frac{1}{n}\sum_{i=1}^{n}\{1-\pi(\bfX_{i})\}\left[\sum_{j=1}^{n}\{\mu_{1}(\bfX_{j})-\mu_{1}(\bfX_{i})\}1(A_{j}=1)\kappa_{1,h_{1},j}(B_{1}^{\T}\bfX_{i})\right]\omega_{h,i}(\bfu;B)\nonumber \\
 & +\frac{1}{n}\sum_{i=1}^{n}\{\pi(\bfX_{i})-A_{i}\}\omega_{h,i}(\bfu;B)\sum_{j=1}^{n}\{Y_{j}-\mu_{1}(\bfX_{i})\}1(A_{j}=1)\kappa_{1,h_{1},j}(B_{1}^{\T}\bfX_{i})\nonumber \\
\stackrel{\triangle}{=}\enspace & J_{0}+J_{1}+J_{2}+J_{3}.\label{eq:DhatZ_mu_1}
\end{align}
To bound $J_{1}$, we re-write it as 
\begin{align*}
J_{1}=\frac{1}{n}\sum_{i=1}^{n}\varepsilon_{1,i}\left\{ \sum_{j=1}^{n}\{1-\pi(\bfX_{j})\}\omega_{h,j}(\bfu;B)\kappa_{1,h_{1},j}(B_{1}^{\T}\bfX_{i})-\{1-\pi(\bfX_{i})\}\omega_{h,i}(\bfu;B)\right\} .
\end{align*}
Since $\E(\varepsilon_{1,i}\cond\bfX_{i})=0$, we can show that $J_{1}$
is a degenerate U-process indexed by $(\bfu,B)$. An application of
Theorem 6 in \citet{NolanP1987U} ensures that $\E(\sup_{(\bfu,B)}\vert J_{1}\vert)\leq C/(n^{2}h_{1}^{d_{1}}h^{d})$.
Thus, by selecting $h_{1}$ in an optimal rate $O\{n^{-1/(2q_{1}+d_{1})}\}$
and coupled with Conditions A3 and A6, we have 
\begin{align}
\sup_{(\bfu,B)}\vert J_{1}\vert=o_{\bbP}\left\{ h^{q}+\left(\frac{\log n}{nh^{d}}\right)^{1/2}\right\} .\label{eq:J2}
\end{align}
Second, similar to the proofs in \citet{HuangC2017Effective}, standard
arguments in kernel smoothing estimation show that 
\begin{align*}
 & \sup_{i}\vert\sum_{j=1}^{n}\{\mu_{1}(\bfX_{j})-\mu_{1}(\bfX_{i})\}1(A_{j}=1)\kappa_{1,h_{1},j}(B_{1}^{\T}\bfX_{i})\vert\\
 & =O_{\bbP}\left\{ h_{1}^{q_{1}}+\left(\frac{\log n}{nh_{1}^{d_{1}}}\right)^{1/2}\right\} =O_{\bbP}\{n^{-q_{1}/(2q_{1}+d_{1})}\}
\end{align*}
by selecting $h_{1}$ in an optimal rate $O\{n^{-1/(2q_{1}+d_{1})}\}$.
Under Conditions A3 and A6, one can further show that this rate is
$o_{\bbP}[h^{q}+\{\log n/(nh^{d})\}^{1/2}]$. Thus, we have 
\begin{align}
\sup_{(\bfu,B)}\vert J_{2}\vert=o_{\bbP}\left\{ h^{q}+\left(\frac{\log n}{nh^{d}}\right)^{1/2}\right\} .\label{eq:J1}
\end{align}
Finally, note that $J_{3}$ is also a degenerate U-process indexed
by $(\bfu,B)$. Thus, by the same argument for $J_{1}$, we can show
that 
\begin{align}
\sup_{(\bfu,B)}\vert J_{3}\vert=o_{\bbP}\left\{ h^{q}+\left(\frac{\log n}{nh^{d}}\right)^{1/2}\right\} .\label{eq:J3}
\end{align}
By substituting (\ref{eq:J2})--(\ref{eq:J3}) into (\ref{eq:DhatZ_mu_1}),
we then have 
\begin{align}
\sup_{(\bfu,B)}\vert I_{1}-\frac{1}{n}\sum_{i=1}^{n}(1-A_{i})\varepsilon_{1,i}\omega_{h,i}(\bfu;B)\vert=o_{\bbP}\left\{ h^{q}+\left(\frac{\log n}{nh^{d}}\right)^{1/2}\right\} .\label{eq:I1}
\end{align}
Following the same arguments above, we can also show that 
\begin{align}
\sup_{(\bfu,B)}\vert I_{2}-\frac{1}{n}\sum_{i=1}^{n}A_{i}\varepsilon_{0,i}\omega_{h,i}(\bfu;B)\vert=o_{\bbP}\left\{ h^{q}+\left(\frac{\log n}{nh^{d}}\right)^{1/2}\right\} .\label{eq:I2}
\end{align}
Substituting (\ref{eq:I1})--(\ref{eq:I2}) into (\ref{eq:DhatZ})
completes the proof. 
\end{proof}
Now we derive the independent and identically distributed representations
of $\widehat{\tau}(B^{\T}\bfx;B)-\tau^{[0]}(\bfx;B)$ and $\partial_{\vctl(B)}\widehat{\tau}(B^{\T}\bfx;B)-\tau^{[1]}(\bfx;B)$.

\begin{lemma} Suppose that Assumption 1 and Conditions A1--A6 are
satisfied. Then, 
\begin{align}
 & \sup_{(\bfx,B)}\vert\widehat{\tau}(B^{\T}\bfx;B)-\tau^{[0]}(\bfx;B)-\frac{1}{n}\sum_{i=1}^{n}\eta_{h,i}^{[0]}(\bfx;B)\vert=o_{\bbP}\left(h^{2q}+\frac{\log n}{nh^{d}}\right),\label{eq:tau0_iid}\\
 & \sup_{(\bfx,B)}\Vert\partial_{\vctl(B)}\widehat{\tau}(B^{\T}\bfx;B)-\tau^{[1]}(\bfx;B)-\frac{1}{n}\sum_{i=1}^{n}\eta_{h,i}^{[1]}(\bfx;B)\Vert=o_{\bbP}\left(h^{2q}+\frac{\log n}{nh^{d+1}}\right),\label{eq:tau1_iid}
\end{align}
where 
\begin{align*}
\eta_{h,i}^{[0]}(\bfx;B) & =\frac{\xi_{i}(\bfx;B)}{f_{B^{\T}\bfX}(B^{\T}\bfx)}\mathcal{K}_{q,h}(B^{\T}\bfX_{i}-B^{\T}\bfx),\\
\eta_{h,i}^{[1]}(\bfx;B) & =\frac{\xi_{i}(\bfx;B)}{f_{B^{\T}\bfX}(B^{\T}\bfx)}\partial_{\vctl(B)}\mathcal{K}_{q,h}(B^{\T}\bfX_{i}-B^{\T}\bfx)\\
 & \quad-\tau^{[1]}(\bfx;B)\mathcal{K}_{q,h}(B^{\T}\bfX_{i}-B^{\T}\bfx)-\frac{f^{[1]}(\bfx,B^{\T}\bfx;B)}{f_{B^{\T}\bfX}(B^{\T}\bfx)}\eta_{h,i}^{[0]}(\bfx;B),
\end{align*}
and $\xi_{i}(\bfx;B)=Z_{i}-\E(Z\cond B^{\T}\bfX=B^{\T}\bfx)$. \end{lemma} 
\begin{proof}
First, (\ref{eq:tau0_iid}) is a direct result of Lemma \ref{lem:tau_Dhat}.
As for (\ref{eq:tau1_iid}), note that 
\begin{multline}
\frac{1}{n}\sum_{i=1}^{n}\widehat{D}_{i}\partial_{\vctl(B)}\mathcal{K}_{q,h}(B^{\T}\bfX_{i}-B^{\T}\bfx)-G^{[1]}(\bfx,B^{\T}\bfx;B)\\
=\frac{1}{n}\sum_{i=1}^{n}\xi_{i}(\bfx;B)\partial_{\vctl(B)}\mathcal{K}_{q,h}(B^{\T}\bfX_{i}-B^{\T}\bfx)+r_{1n}(\bfx;B),\label{eq:dNn}
\end{multline}
where $\sup_{(\bfx,B)}\vert r_{1n}(\bfx,B)\vert=o_{\bbP}[h^{q}+\{\log n/(nh^{d+1})\}^{1/2}]$,
by paralleling the proof steps of Lemma \ref{lem:tau_Dhat}. Now by
using the Taylor expansion, we have 
\begin{align}
 & \partial_{\vctl(B)}\widehat{\tau}(B^{\T}\bfx;B)-\tau^{[1]}(\bfx;B)\nonumber \\
=\enspace & \frac{\sum_{i=1}^{n}\widehat{D}_{i}\partial_{\vctl(B)}\mathcal{K}_{q,h}(B^{\T}\bfX_{i}-B^{\T}\bfx)/n-\tau^{[0]}(\bfx;B)\sum_{i=1}^{n}\partial_{\vctl(B)}\mathcal{K}_{q,h}(B^{\T}\bfX_{i}-B^{\T}\bfx)/n}{f_{B^{\T}\bfX}(B^{\T}\bfx)}\nonumber \\
 & -\frac{\tau^{[1]}(\bfx;B)}{n}\sum_{i=1}^{n}\mathcal{K}_{q,h}(B^{\T}\bfX_{i}-B^{\T}\bfx)-\frac{f^{[1]}(\bfx,B^{\T}\bfx;B)}{f_{B^{\T}\bfX}(B^{\T}\bfx)}\{\widehat{\tau}(B^{\T}\bfx;B)-\tau^{[0]}(\bfx;B)\}\nonumber \\
 & +r_{2n}(\bfx;B),\label{eq:tau1_Taylor}
\end{align}
where 
\begin{eqnarray*}
r_{2n}(\bfx,B) & = & O_{\bbP}\{|\widehat{\tau}(B^{\T}\bfx;B)-\tau^{[0]}(\bfx;B)|^{2}\\
 &  & +\Vert\sum_{i=1}^{n}\widehat{D}_{i}\partial_{\vctl(B)}\mathcal{K}_{q,h}(B^{\T}\bfX_{i}-B^{\T}\bfx)/n-G^{[1]}(\bfx,B^{\T}\bfx;B)\Vert^{2}\}.
\end{eqnarray*}
Finally, substituting the result in Lemma \ref{lem:tau_Dhat} and
(\ref{eq:dNn}) into (\ref{eq:tau1_Taylor}) completes the proof. 
\end{proof}
\begin{corollary}\label{cor:rate} Suppose that Assumption 1 and
Conditions A1--A6 are satisfied. Then, 
\begin{align*}
 & \sup_{(\bfx,B)}\vert\widehat{\tau}(B^{\T}\bfx;B)-\tau^{[0]}(\bfx;B)\vert=O_{\bbP}\left\{ h^{q}+\left(\frac{\log n}{nh^{d}}\right)^{1/2}\right\} ,\\
 & \sup_{(\bfx,B)}\Vert\partial_{\vctl(B)}\widehat{\tau}(B^{\T}\bfx;B)-\tau^{[1]}(\bfx;B)\Vert=O_{\bbP}\left\{ h^{q}+\left(\frac{\log n}{nh^{d+1}}\right)^{1/2}\right\} .
\end{align*}
\end{corollary}

\section{Proofs of Theorems 2 and 3\label{sec:Proofs-of-Theorems}}

\subsection{Proof of Theorem 2}
\begin{proof}
Let $\bar{\tau}^{-i}(B^{\T}\bfX_{i};B)={\sum_{j\neq i}Z_{j}\mathcal{K}_{q,h}(B^{\T}\bfX_{j}-B^{\T}\bfX_{i})}/{\sum_{j\neq i}\mathcal{K}_{q,h}(B^{\T}\bfX_{j}-B^{\T}\bfX_{i})}$.
We can decompose $\textsc{cv}(d,B,h)$ into 
\begin{align*}
\textsc{cv}(d,B,h) & =\frac{1}{n}\sum_{i=1}^{n}\{Z_{i}-\bar{\tau}^{-i}(B^{\T}\bfX_{i};B)\}^{2}+\frac{1}{n}\sum_{i=1}^{n}(\widehat{D}_{i}-Z_{i})^{2}\\
 & +\frac{1}{n}\sum_{i=1}^{n}\{\widetilde{\tau}^{-i}(B^{\T}\bfX_{i};B)-\bar{\tau}^{-i}(B^{\T}\bfX_{i};B)\}^{2}\\
 & +\frac{2}{n}\sum_{i=1}^{n}(\widehat{D}_{i}-Z_{i})\{\widetilde{\tau}^{-i}(B^{\T}\bfX_{i};B)-\bar{\tau}^{-i}(B^{\T}\bfX_{i};B)\}\\
 & +\frac{2}{n}\sum_{i=1}^{n}(\widehat{D}_{i}-Z_{i})\{Z_{i}-\tau(\bfX_{i})\}+\frac{2}{n}\sum_{i=1}^{n}(\widehat{D}_{i}-Z_{i})\{\tau(\bfX_{i})-\bar{\tau}^{-i}(B^{\T}\bfX_{i};B)\}\\
 & +\frac{2}{n}\sum_{i=1}^{n}\{Z_{i}-\tau(\bfX_{i})\}\{\widetilde{\tau}^{-i}(B^{\T}\bfX_{i};B)-\bar{\tau}^{-i}(B^{\T}\bfX_{i};B)\}\\
 & +\frac{2}{n}\sum_{i=1}^{n}\{\tau(\bfX_{i})-\bar{\tau}^{-i}(B^{\T}\bfX_{i};B)\}\{\widetilde{\tau}^{-i}(B^{\T}\bfX_{i};B)-\bar{\tau}^{-i}(B^{\T}\bfX_{i};B)\}\\
 & \stackrel{\triangle}{=}SS_{1}+SS_{2}+SS_{3}+SC_{1}+SC_{2}+SC_{3}+SC_{4}+SC_{5}.
\end{align*}
Note that 
\begin{equation}
\sup_{i}\vert\widehat{D}_{i}-Z_{i}\vert\leq\sum_{a=0}^{1}\sup_{(\bfu,B)}\vert\widehat{\mu}_{a}(\bfu;B)-\mu_{a}(\bfu;B)\vert=o_{\bbP}\left\{ h^{q}+\left(\frac{\log n}{nh^{d}}\right)^{1/2}\right\} ,\label{eq:DhatZ_rate}
\end{equation}
\begin{multline}
\sup_{(i,B)}\vert\widetilde{\tau}^{-i}(B^{\T}\bfX_{i};B)-\bar{\tau}^{-i}(B^{\T}\bfX_{i};B)\vert\leq C\sum_{a=0}^{1}\sup_{(\bfu,B)}\vert\widehat{\mu}_{a}(\bfu;B)-\mu_{a}(\bfu;B)\vert\\
=o_{\bbP}\left\{ h^{q}+\left(\frac{\log n}{nh^{d}}\right)^{1/2}\right\} \label{eq:taudiff_rate}
\end{multline}
for some positive constant $C$, by using Conditions A1--A3, Condition
A6, and standard arguments in kernel smoothing estimation.

When $\csp(B)\supseteq\csp(B_{{\tau}})$, Theorem 1 of \citet{HuangC2017Effective}
implies that $SS_{1}=\sigma_{{\tau}}^{2}+O_{\bbP}\{h^{2q}+\log n/(nh^{d})\}$,
where $\sigma_{{\tau}}^{2}=\E[\{Z-\tau(\bfX)\}^{2}]$. From (\ref{eq:DhatZ_rate})--(\ref{eq:taudiff_rate}),
$\sup_{B}\vert SS_{3}\vert$ and $\sup_{B}\vert SC_{1}\vert$ are
of order $o_{\bbP}\{h^{2q}+\log n/(nh^{d})\}$. Further, by using
$\sup_{(\bfx,B)}\vert\bar{\tau}(B^{\T}\bfx;B)-\tau(\bfx)\vert=O_{\bbP}[h^{q}+\{\log n/(nh^{d})\}^{1/2}]$,
$\sup_{B}\vert SC_{3}\vert$ and $\sup_{B}\vert SC_{5}\vert$ are
also of order $o_{\bbP}\{h^{2q}+\log n/(nh^{d})\}$. Now note that
$SC_{4}$ can be expressed a U-process indexed by $B$ asymptotically.
By using the same proof steps for the cross term in Theorem 1 of \citet{HuangC2017Effective},
one can immediately conclude that $\sup_{B}\vert SC_{4}\vert=o_{\bbP}\{h^{2q}+\log n/(nh^{d})\}$.
Combining the results above, we have $\textsc{cv}(d,B,h)=SS_{1}+SS_{2}+SC_{2}+o_{p}(SS_{1})$
uniformly in $B$. When $\csp(B)\nsupseteq\csp(B_{{\tau}})$, Theorem
1 of \citet{HuangC2017Effective} implies that $SS_{1}=\sigma_{{\tau}}^{2}+b_{{\tau}}^{2}(B)+o_{\bbP}(1)$.
By using (\ref{eq:DhatZ_rate})--(\ref{eq:taudiff_rate}) again,
we have $\textsc{cv}(d,B,h)=SS_{1}+SS_{2}+SC_{2}+o_{\bbP}(1)$ uniformly
in $B$. Finally, since $SS_{2}$ and $SC_{2}$ are independent of
$B$, the minimizer of $\textsc{cv}(d,B,h)$ has the same asymptotic
distribution as the minimizer of $SS_{1}$. Thus, Theorem 2 is a direct
result of Theorem 2 in \citet{HuangC2017Effective}. 
\end{proof}

\subsection{Proof of Theorem 3}
\begin{proof}
By using first-ordered Taylor expansion, we have 
\begin{align*}
\widehat{\tau}(\widehat{B}^{\T}\bfx;\widehat{B})-\tau(\bfx) & =\widehat{\tau}(\widehat{B}^{\T}\bfx;\widehat{B})-\widehat{\tau}(B_{{\tau}}^{\T}\bfx;B_{{\tau}})+\widehat{\tau}(B_{{\tau}}^{\T}\bfx;B_{{\tau}})-\tau(\bfx)\\
 & =\partial_{\vctl(B)}\widehat{\tau}(\bar{B}^{\T}\bfx;\bar{B})\vctl(\widehat{B}-B_{{\tau}})+\widehat{\tau}(B_{{\tau}}^{\T}\bfx;B_{{\tau}})-\tau(\bfx),
\end{align*}
where $\bar{B}$ lies on the line segment between $\widehat{B}$ and
$B_{{\tau}}$. From Theorem 2, $\vctl(\widehat{B}-B_{{\tau}})=O_{\bbP}(n^{-1/2})$.
Coupled with Corollary \ref{cor:rate} and continuous mapping theorem,
$\partial_{\vctl(B)}\widehat{\tau}(\bar{B}^{\T}\bfx;\bar{B})=O_{\bbP}(1)$.
Moreover, from (\ref{eq:tau0_iid}), we have 
\begin{align*}
(nh_{\tau}^{d_{{\tau}}})^{1/2}\{\widehat{\tau}(B_{{\tau}}^{\T}\bfx;B_{{\tau}})-\tau(\bfx)\}-h_{\tau}^{q_{\tau}}\gamma(\bfx)\to{\rm N}\{0,\sigma_{\tau}^{2}(\bfx)\}
\end{align*}
in distribution as $n\to\infty$. Combining the results above completes
the proof of Theorem 3. 
\end{proof}

\end{document}